\documentclass[10pt]{article}
\usepackage{amsmath,amssymb,amsfonts,amsthm}
\usepackage[margin=1.2in]{geometry}
\usepackage{algorithmic}
\usepackage{graphicx}
\usepackage{subfig}
\usepackage{textcomp}
\usepackage[dvipsnames]{xcolor}
\usepackage{caption}
\usepackage{hyperref}
\usepackage{url}

\usepackage{authblk}

\graphicspath{{Figures/}}
\graphicspath{{FiguresSubmitted/}}
\providecommand{\keywords}[1]
{
  \small	
  \textbf{\textit{Keywords---}} #1
}
\DeclareMathOperator*{\poa}{PoA}

\newtheoremstyle{teo}
  {3pt} 
  {3pt} 
  {\itshape} 
  {} 
  {\bfseries} 
  {.} 
  {.5em} 
  {} 

\theoremstyle{teo}
\newtheorem{lem}{Lemma}
\newtheorem{prop}{Proposition}
\newtheorem{corol}{Corollary}
\newtheorem{thm}{Theorem}

\newtheoremstyle{defi}
  {3pt} 
  {3pt} 
  {} 
  {} 
  {\bfseries} 
  {.} 
  {.5em} 
  {} 

\theoremstyle{defi}
\newtheorem{defi}{Definition}
\newtheorem{assum}{Assumption}

\newtheoremstyle{rema}
  {3pt} 
  {3pt} 
  {} 
  {} 
  {\bfseries} 
  {.} 
  {.5em} 
  {} 
\theoremstyle{rema}
\newtheorem{rem}{Remark}

\usepackage{color}

\date{}
\begin{document}
\title{Potential detrimental effects of real-time route recommendations in traffic networks}
\author[1]{Tommaso Toso}
\author[1]{Paolo Frasca}
\author[1,2]{Alain Y. Kibangou}

\affil[1]{Univ.\ Grenoble Alpes, CNRS, Inria, Grenoble INP, GIPSA-lab, 38000 Grenoble, France (e-mail: name.surname@grenoble-inp.fr).}
\affil[2]{Univ.\ of Johannesburg (Auckland Park Campus), Johannesburg 2006, South Africa.}
\renewcommand{\Authfont}{\large}
\renewcommand{\Affilfont}{\scriptsize}
\maketitle

\begin{abstract}

Navigation apps 
have become pervasive in providing real-time route recommendations to travelers willing to minimize their travel times. However, such technologies introduce new complexities, raising concerns about their overall impact on traffic networks. This paper focuses on evaluating the effect of navigation apps on traffic flows, particularly examining how real-time route recommendations influence network efficiency and congestion. Using a dynamical network flow model, we study traffic dynamics between an origin-destination pair, where a fraction of drivers follow app recommendations while others rely on fixed route preferences. By incorporating supply-demand mechanisms to account for capacity and volume constraints on routes, we uncover \emph{partial demand transfer}, i.e., only a portion of the traffic demand is able to traverse the network, while the rest builds up congestion at the origin. We prove that the dynamics converges to a globally stable equilibrium and we provide a detailed analysis of this equilibrium when the choices of the informed drivers follow a logit model, correlating the emergence of partial demand transfer to the penetration rate of navigation apps among users.
\end{abstract}


\section{Introduction}


Transportation networks are crucial to modern economies, facilitating the movement of people and goods. Despite their importance and the need for efficiency, these networks often suffer from congestion, which affects not only the economy but also quality of life, environmental sustainability, and social equity.

Over the past two decades, innovations in information and communication technologies, along with the widespread adoption of GPS-enabled devices and smartphones, have led to the emergence of new digital services aimed at simplifying mobility for users, exploiting real-time information. Navigation apps are a prime example of these mobility services. Designed to assist travelers in finding the most efficient routes, these apps utilize GPS technology, digital maps, and live traffic data to provide optimized routing recommendations that minimize travel time and avoid slowdowns and bottlenecks. The near-ubiquitous adoption of navigation apps is evident in their impressive user bases, with Google Maps having over 2 billion monthly users worldwide and Gaode Map reaching 730 million monthly active users in China as of 2022 \cite{stat;maps,stat;nav}.

Recent research has explored how navigation apps are reshaping mobility and their impact on traffic efficiency across networks. While these technologies are aimed to improve travelers' choices, they also introduce complexities that can lead to potential inefficiencies. The key question is whether the presence of informed travelers, following navigation app recommendations, improves overall traffic conditions or instead their actions create negative externalities and unexpected consequences. This paper contributes to this research line, focusing on evaluating how navigation apps impact road traffic and their potential to degrade network efficiency.



In recent years, the impact of navigation apps on traffic networks has been commonly modeled through selfish routing, where users act as self-interested agents seeking to minimize their individual travel times. This approach aligns well with the behavior of app-informed users who follow the fastest routes recommended by these apps. The inefficiencies arising from selfish and uncoordinated behaviors have been widely studied, particularly through the concepts of Wardrop equilibria \cite{wardrop1952}, which describe the resulting traffic flows, and Price of Anarchy, a common measure used to quantify the inefficiency of such equilibria \cite{poa,rough}. Two main frameworks have emerged in order to assess the impact of app-informed users: heterogeneous routing games \cite{dafermos72} and dynamical network flow models \cite{surveycomo}. Heterogeneous routing games allowed for assessing how the presence of app-informed users influences the Price of Anarchy, linking the widespread use of navigation apps with a reduction in overall traffic efficiency across networks \cite{thai,ibp,th:cabannes}, while dynamical network flow models focus on the stability of Wardrop equilibria, examining whether traffic networks converge to these equilibria \cite{como2022,bianchin2024}. Despite their value in understanding traffic dynamics, both frameworks present significant limitations—notably, their lack of capacity and volume constraints on network links.

The main goal of this paper is to address these limitations by explicitly modeling capacity and volume constraints in order to capture more realistic network dynamics and uncover potential inefficiencies or unintended consequences that may arise from the interplay between user behavior, digital platforms, and network limitations. We study the problem in the context of dynamical network flow models, using an ODE system to describe traffic dynamics between an origin-destination pair subject to a traffic demand of vehicles that have to choose one of two possible travel options. A fraction of the users, namely the app-informed users, rely on real-time route recommendations from a navigation app to determine the shortest travel time route. The remaining drivers select routes based on predetermined preferences. The fraction of traffic demand allocated to each route, known as the \emph{routing ratio}, is a dynamic variable that adjusts in response to traffic density. This feedback mechanism dynamically regulates flow based on traffic conditions, reflecting app-informed users’ responsiveness to real-time travel times. Within each route, the traffic dynamics integrate a supply-demand mechanism that imposes capacity and volume constraints on the two available routes. Thanks to the supply-and-demand mechanism, we can identify a critical consequence of navigation app usage that goes beyond the well-known issue of reduced traffic efficiency due to increased total travel time. This consequence is known as \emph{partial demand transfer}, where the system converges to an equilibrium flow on the two routes, allowing only a portion of the exogenous traffic to enter and traverse the network. The remaining traffic gets stuck at the origin, causing congestion to build up at the entrance.


\subsubsection*{Contribution} 
Towards assessing the impact of navigation apps on traffic efficiency, we claim three main contributions.
First, we propose an original dynamical network flow model, which is able to capture the strategic behavior of app-informed drivers in response to variations in the travel times on the two routes and its impact on the traffic state. Our model is the first to incorporate a supply and demand mechanism on network links and state-dependent routing ratios. Second, we demonstrate global asymptotic stability for a broad family of user preference dynamics, i.e., routing ratios. Third, we study the properties of the unique equilibrium assuming that user preferences follow the logit choice model \cite{sandholm}. This analysis is performed in two limit regimes. In the regime of high compliance to app’s recommendations, we show that the equilibrium approximates the Wardrop equilibrium of a suitable routing game. In the low compliance regime, we derive a linear approximation of the user preferences' dynamics. This twofold study shows that navigation apps can degrade the network efficiency, by increasing the average travel time (in line with previous works \cite{ibp,th:cabannes}) and by leading to partial demand transfer. The latter drawback was not identified by previous work. The key variable in our steady-state analysis is the \emph{penetration rate}, that is, the share of app-informed users in the total demand. Our analysis shows that a high penetration rate is likely to degrade the network efficiency when compliance is high.

Finally, the theoretical results are showcased through numerical experiments on our model and corroborated by simulations using the microscopic traffic simulator Aimsun \cite{aimsun}, offering comprehensive validation of this work.

\subsubsection*{Related works}
Heterogeneous routing games have been the main model for studying the relationship between the penetration rate of navigation apps among users and traffic efficiency, as they provide a setting to naturally distinguish users between informed and uninformed. In \cite{thai}, the authors highlighted two key findings: increasing the penetration rate can relief congestion on major roads, but this leads to increased traffic on secondary roads, a phenomenon known as cut-through traffic\footnote{Cut-through traffic means traffic that passes through a given residential neighborhood that has neither an origination nor destination point in that neighborhood.}. Contradicting the insight that more informed-users always improves efficiency, \cite{ibp} introduced the concept of Informational Braess' Paradox (IBP), showing that increasing users' knowledge of available routes can increase total travel time. Interestingly, the highest inefficiency arises when all users have complete route information. Building on these insights, Cabannes \cite{th:cabannes} explored a similar two-class game and demonstrated that higher penetration of navigation apps can push traffic toward Wardrop equilibria, potentially worsening total travel time and intensifying cut-through traffic on secondary roads. 

Macroscopic PDE models have also been used to study the effects of selfish routing driven by real-time information. These models combine mass conservation with a Hamilton-Jacobi equation to represent traffic dynamics and user choices. Studies \cite{goatin1,goatin2,cristiani} provide numerical experiments aligning with those stemming from their static counterparts.

The observation that higher penetration of navigation apps can compromise network efficiency aligns with findings from studies on Bayesian routing games. In this setting, users have incomplete or uncertain knowledge about the traffic state, and the goal is to determine how to best inform users to optimize network performance. These works reveal that providing users with full information is often not the optimal strategy. Instead, carefully limiting the information available to users can lead to better overall outcomes by preventing behaviors that exacerbate congestion and inefficiency across the network \cite{bryce2024,wu,mlifac}.

The stability of dynamical network flows has been extensively studied in the literature \cite{surveycomo,como2013a,como2013b,como2013c,como2015,dahleh2018,bianchin2024}, often relying on monotonicity \cite{Hirsch} and contractivity \cite{bullo} properties to demonstrate asymptotic stability. However, these studies typically overlook supply and demand mechanisms that regulate capacity constraints on network links. Unlike traffic networks, these models assume unconstrained link inflows, which can be arbitrarily high but are constrained only by link capacities for outflows. This setup ensures that traffic dynamically adjusts to minimize travel times, facilitating complete transferability of exogenous flows up to the network's min-cut capacity. In contrast, our model incorporates a realistic supply and demand
mechanism on each link. This addition fundamentally alters the traffic dynamics and disrupts the contractivity property. Supply and demand models are also present in \cite{coogan2015,lovisari2014}. However, in those works, user preferences are kept fixed and do not evolve according to the state of the network.

\subsubsection*{Conference version} A preliminary and incomplete version of this work was presented in \cite{tap}. The main improvements in this journal version include: (i) a more general model, allowing for a fraction of app-informed users and more flexible routing ratios; (ii) a new stability analysis based on the system's monotonicity; (iii) an analysis of logit-based routing ratios; and (iv) validation through simulations conducted in Aimsun.

\subsubsection*{Paper organization}
The considered model and the most relevant notions are described in Section~\ref{sec2}  before carrying out the stability analysis in Section~\ref{sec3}. In Section~\ref{sec4}, we study the properties of the equilibrium for the case of logit routing ratios. This analysis illustrates the impact of routing recommendations. In Section~\ref{sec5}, we propose both macroscopic and microscopic numerical simulations that illustrate and corroborate our theoretical results. Finally, Section~\ref{sec6} presents some concluding remarks and ideas for future work. 

\section{Model description}\label{sec2}

\begin{figure}
    \centering    
\includegraphics[width=.8\textwidth]{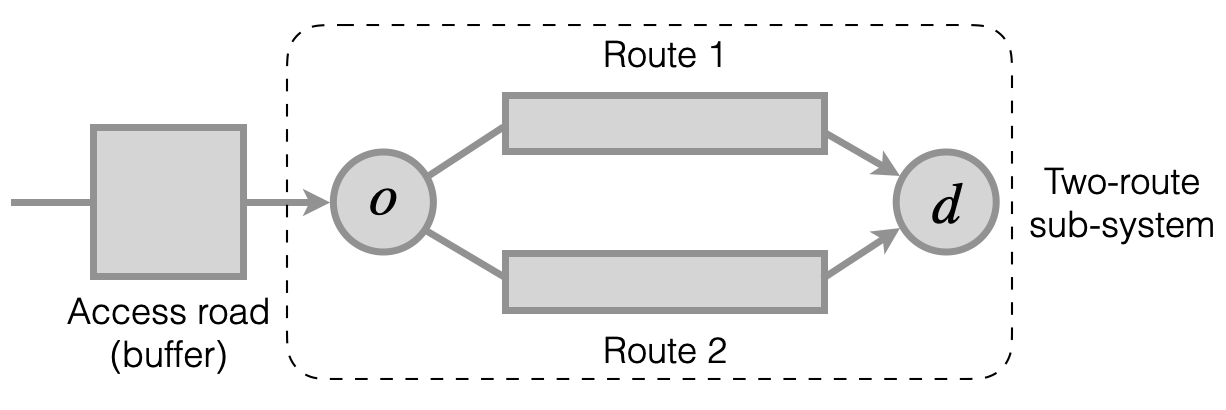}
    \caption{The origin-destination pair $\mathcal{G}$.}
    \label{grp}
\end{figure}
In this section, we introduce the dynamical network flow model we desgined to tackle the problem.
Consider an origin-destination pair $\mathcal{G}$ with two possible commute options (syn. routes). For each route~$l$, $l=1,2$, the positive values $x_l$, $B_l$, $C_l$, $F_l$ and $L_l$ represent the density (veh/km), jam density (veh/km), critical density (veh/km), capacity (veh/h) and length (km), respectively, with \mbox{$C_l<B_l,\ l=1,2$}. The two-route system is subject to an exogenous traffic flow, which is supplied via an access road A (see Figure~\ref{grp}). The traffic demand $\Phi$ is split between the two routes according to the routing ratios $R_1$, $R_2$, corresponding to the fraction of exogenous flow directed
toward Route $1$ and Route $2$, respectively. The traffic dynamics over the network is
captured by the following ODE system, consisting in a conservation law describing how the density $x_l\in[0,B_l]$ of Route $l$ evolves over time:
\begin{equation}
\begin{aligned}
    &\dot{x}_l=\frac{1}{L_l}\left(f_l^\mathrm{in}-f_l^\mathrm{out}\right),\quad l=1,2,\\
    &f_l^\mathrm{in}=\min\{\Phi R_l(\tau(x)),s_l(x_l)\},\\
    &f_l^\mathrm{out}=d_l(x_l)
\end{aligned}
    \label{ch:dynselfish:os}
\end{equation}
where $x=(x_1,x_2)^T$ in \mbox{$\Omega:=[0,\overline{x}_1]\times[0,\overline{x}_2]$}.
Equation \eqref{ch:dynselfish:os} states that the instant variation of traffic density on a route equals the difference between its inflow and its outflow. Specifically, the inflow is the minimum between the fraction of exogenous flow aiming to enter the link, and the supply $s_l$ of the link, which represent the maximum inflow that can be accommodated by the link given its current traffic density. Supply is modeled as a decreasing, Lipschitz continuous function of the link density. Its definition captures the fact that when the route density exceeds its critical threshold, the route capacity progressively shrinks:
\begin{equation}  s_l(x_l):=F_l\cdot\min\biggl\{1,\frac{B_l-x_l}{B_l-C_l}\biggl\}\quad l=1,2.
    \label{sup}
\end{equation}
The link outflow corresponds to the route demand, which represents the flow attempting to leave the link given the current density. The route demand is modeled as an increasing Lipschitz continuous function of the density, saturating at the route's capacity when the critical density is reached:
\begin{equation}
    d_l(x):=\min\{v_lx_l,F_l\},\:\:\mathrm{with} \quad v_l=F_l/C_l\:\:\mathrm{(km/h)},\:\: l=1,2.
    \label{dem}
\end{equation}
This supply and demand mechanism inspired by Daganzo’s cell transmission model \cite{daganzo1994,daganzo1995} acts as capacity and volume constraint on the link, regulating the flow that enters and leaves the link according to its traffic density.

The access road to the network also presents a traffic dynamics, captured by the following equation that describes how the density $x_A\in[0,+\infty)$ evolves over time:
\begin{equation}
    \dot{x}_A=\frac{1}{L_A}\left(\Phi-\sum_{l=1,2}f_l^\mathrm{in}\right),\quad l=1,2,
    \label{buffer}
\end{equation}
The access road does not incorporate the supply and demand mechanisms of the two routes; instead, it functions as a buffer with infinite capacity and jam density, designed to supply the exogenous flow $\Phi$ into Route~1 and Route~2 while capturing potential spillbacks due to congestion in the two routes.


\subsection{Routing ratios and travel times}
Each route is characterised by a {\em travel time} that, consistently the literature~\cite{kachroo2016,trb}, is assumed to be a strictly increasing $C^1$ function $\tau_l:[0,B_l]\rightarrow\mathbb{R}_{>0}$ of density $x_l,\ l=1,2$.
\emph{Routing ratios} $R_l$ quantify the ratio of demand directed toward each route and they are modeled in order
to account for the presence of \emph{informed users}, i.e., users relying on fastest-route recommendations from a some navigation system. Assuming that a fraction of users $\alpha\in(0,1]$, which we refer to as the \textit{penetration rate}, is influenced by the routing recommendations, whereas the remaining fraction $1-\alpha$ splits according to fixed routing ratios $r^0=(r_1^0$, $r_2^0),\ r_1^0+r_2^0=1$, we write routing ratios as follows
\begin{equation}
    R_l(\tau(x)):=(1-\alpha)r_l^0+\alpha r_l(\tau(x)),\quad l=1,2.
    \label{ch:dynselfish:rr}
\end{equation}
The second term in \eqref{ch:dynselfish:rr} captures the behavior of the informed users with respect to travel times, where $r_l(\tau):\mathbb{R}_{>0}^2\rightarrow[0,1]$ is a globally Lipschitz $C^1$ function, $l=1,2$. These functions are assumed to satisfy \mbox{$r_1(\tau)+r_2(\tau)=1,\ \forall \tau\in\mathbb{R}_{>0}^2$}, so that \mbox{$0\leq R_l(\tau)\leq1,\ \forall \tau\in\mathbb{R}_{>0}^2,\ l=1,2$}, and that \mbox{$R_1(\tau)+R_2(\tau)=1$}.

We now focus on a special class of routing ratios, \textit{monotone routing ratios}, originally introduced in~\cite{como2013a}. 
\begin{defi}[Monotone ratios]\label{def1}
    Routing ratios \eqref{ch:dynselfish:rr} are said to be \textit{monotone} if 
    \begin{equation}
        \frac{\partial R_l}{\partial \tau_j}\geq0,\quad l\neq j,\ l=1,2.
        \label{mrr}
    \end{equation}
    If the inequality in \eqref{mrr} is strict, then they are \textit{strictly monotone}.
\end{defi}
Monotone routing ratios ensure that the higher the travel time on a route, the fewer informed users are directed toward it. This dependence captures the fact that app-informed users seek to minimize their travel time. 
Since travel times are strictly increasing in the traffic densities, for monotone routing ratios we have 
    \begin{equation}
        \frac{\partial R_l(\tau(x))}{\partial x_j}\geq0,\quad
        \frac{\partial R_l(\tau(x))}{\partial x_l}\leq0,\quad l\neq j,\ l=1,2.
    \end{equation}

In what follows, we assume:
\begin{assum}[Strict monotonicity]
    Routing ratios \eqref{ch:dynselfish:rr} are strictly monotone.
    \label{ch:dynselfish:ass4}
\end{assum}
A well-established model of routing ratios, which falls under this assumption, is the logit routing ratios:

\begin{equation}
    R_l(\tau)=(1-\alpha)r_l^0+\frac{\alpha}{1+\frac{r_j^0}{r_l^0}\exp\left(\frac{1}{\eta}\left(\tau_l-\tau_j\right)\right)},
\label{lrr}
\end{equation}
$l\neq j,\ l=1,2$. 
The state-dependent term takes the form of the \textit{logit choice model} \cite{sandholm,daganzo1977,akiva}, where $\eta>0$ is the so-called {\em noise} parameter. When $\eta$ approaches zero, the logit model approximates a best response dynamics. The noise parameter can be interpreted as measuring how reliable or how {\em accurate} is the travel-time information provided by the navigation app.
In this work, we mainly interpret the parameter $1/\eta>0$ as measuring the \textit{users' compliance}. 
When $1/\eta\rightarrow 0$, i.e., users' compliance is very low, users do not really exploit the information and the demand splitting stays close to $r^0$. On the contrary, when $1/\eta\rightarrow +\infty$, all users tend to take the route with shortest travel time.  

\subsection{Unsatisfied demand and congestion}

Supply and demand naturally allows to define states that exhibit \textit{unsatisfied demand} and \textit{congestion}. 
When the fraction of traffic demand $\Phi R_l(x)$ directed to a route is less than the route supply $s_l(x)$, then the demand can enter the route freely; the demand is satisfied (S). On the contrary, if $\Phi R_l(x)$ exceeds $s_l(x)$, then the first term in \eqref{td} gets saturated. In this case, the demand is unsatisfied (U). Differently, if the internal demand $d_l(x)$ of the route, i.e., the quantity of drivers that aims at leaving the route, is less than its capacity $F_l$, then the route is in free-flow (F), otherwise it is congested (C). 
Therefore, each route is characterized by four possible \textit{route modes} (see Table~\ref{tab1}), allowing to rewrite \eqref{ch:dynselfish:os}-\eqref{buffer} as follows:  
\begin{equation}
    \dot{x}_l=\begin{cases}
    \dfrac{1}{L_l}\Phi R_l(\tau(x))-v_lx_l, & \text{if } x_l\leq C_l,\ \Phi R_l(\tau(x))\leq s_l(x),\quad \text{ SF}
    \\ \\
    \dfrac{1}{L_l}\left(F_l-v_lx_l\right), & \text{if }x_l\leq C_l,\ \Phi R_l(\tau(x))>s_l(x),\quad \text{ UF}
    \\ \\
    \dfrac{1}{L_l}\left(\Phi R_l(\tau(x))-F_l\right), & \text{if }x_l> C_l,\ \Phi R_l(\tau(x))\leq s_l(x),\quad \text{ SC}
    \\ \\
     \dfrac{F_l(C_l-x_l)}{L_l(B_l-C_l)}, & \text{if }x_l> C_l,\ \Phi R_l(\tau(x))>s_l(x),\quad \text{ UC}
    \end{cases},\quad l=1,2.
    \label{td}
\end{equation}
In the case where the fraction of flow directed toward a route exceeds its supply, then the exogenous flow $\Phi$ is \emph{partially transferred}, and the excess flow accumulates in the buffer A. In order to avoid discussing uninteresting cases, we shall assume that demand does not exceed the network capacity, that is, it is possible to divide the demand so as to satisfy it completely.
\begin{assum}[Satisfiable traffic demand] 
Traffic demand is such that
    \begin{equation}
    \Phi<F_1+F_2.
    \end{equation}
    \label{ass1}
\end{assum}

\begin{table}
\caption{Notation for the four route modes in \eqref{td}. }
\begin{center}
\begin{tabular}{@{}lll@{}}
\hline
   & Route demand     & Traffic regime      \\ \hline 
SF & satisfied & free-flow \\
UF & unsatisfied & free-flow \\
SC & satisfied & congested \\
UC & unsatisfied & congested \\ 
\hline
\end{tabular}
\label{tab1}
\end{center}
\end{table}

\section{Equilibria and stability analysis}\label{sec3}
The following section is devoted to analyzing the behavior of the dynamical system \eqref{ch:dynselfish:os}-\eqref{ch:dynselfish:rr} introduced in Section~\ref{sec2}, aiming to derive insights into the influence of informed users on traffic dynamics. This involves examining the system's asymptotic behavior.

Existence and uniqueness of the solutions of \eqref{ch:dynselfish:os}-\eqref{ch:dynselfish:rr} are ensured by the fact that the system is Lipschitz continuous. 
We remark that \eqref{ch:dynselfish:os}-\eqref{ch:dynselfish:rr} is well-posed with respect to $\Omega\times[0,+\infty)$, i.e., $\Omega\times[0,+\infty)$ is positively invariant. Indeed, the vector field is always pointing inward on the boundaries of $\Omega\times[0,+\infty)$. 
Hence, if $x_0\in\Omega\times[0,+\infty)$, then $\gamma(t,x_0)\in\Omega\times[0,+\infty)$, $\forall t\geq0$.
In the following, we perform a stability analysis of the system. It is important to note that the dynamics of the two routes are entirely unaffected by the behavior of the access road buffer A. Hence, the overall system behavior is exclusively determined by the traffic dynamics of the two routes. For this reason, our analysis will focus solely on the two-route subsystem, and from this analysis, we will eventually deduce the dynamics of the access road A.

Now, observe that \eqref{ch:dynselfish:os}-\eqref{dem} and \eqref{ch:dynselfish:rr} is a state-dependent switched system, where each \textit{system mode} can be defined as a route combination mode $\text{M}_1\text{-M}_2$, where $\text{M}_1,\text{M}_2\in\{\mathrm{SF},\mathrm{UF},\mathrm{SC},\mathrm{UC}\}$ indicate the modes of Route~$1$ and Route~$2$, respectively. 
We will refer to the sub-system associated with system mode $\text{M}_1\text{-M}_2$ with the notation $\Sigma^{\text{M}_1\text{-M}_2}$ and we will indicate as $\Omega^{\text{M}_1\text{-M}_2}$ the open region of the state space where sub-system $\Sigma^{\text{M}_1\text{-M}_2}$ is active. Notice that $\Omega^\text{UF-UF}$ will always be empty because of Assumption \ref{ass1}. 

We are now going to present some preliminary results that allow us to simplify the stability analysis of \eqref{ch:dynselfish:os}-\eqref{dem} and \eqref{ch:dynselfish:rr}. Before presenting them, let us define the two following regions:
\begin{equation}
    P:=\{x\in\Omega|\ 0\leq x_1\leq C_1,\,0\leq x_2\leq C_2\},\ \ Q:=\Omega\setminus P.
\end{equation}
Notice that $P=\overline{\Omega}^{\text{SF-SF}}\cup \overline{\Omega}^{\text{UF-SF}}\cup \overline{\Omega}^{\text{SF-UF}}$, and that all regions $\Omega^{\text{M}_1\text{-M}_2}$ such that $\text{M}_1\in\{\text{SC},\text{UC}\}$ or $\text{M}_2\in\{\text{SC},\text{UC}\}$ are contained in $Q$. Given a domain $\mathcal{D}\subseteq\mathbb{R}^d$ and a system of differential equations $\dot{y}=g(y)$ with \mbox{$g:D\rightarrow \mathbb{R}^d$} and having a unique solution $\varphi(t,y_0), t\geq0$ from each initial condition $y_0\in\mathcal{D}$, we say that $\mathcal{Y}\subseteq\mathcal{D}$ is \textit{positively invariant} if $y_0\in\mathcal{Y}$ implies that $\varphi(t,y_0)\in\mathcal{Y},\ \forall t\geq0$. We say that $\mathcal{Y}$ is \textit{globally attractive} if, for any open neighbourhood $U$ of $\mathcal{Y}$, $\forall y_0\in\mathcal{D}$, there exists a time $\tau (y_0)>0$ such that \mbox{$\varphi(t,y_0)\in U ,\forall t >\tau (y_0)$}.

\begin{lem}[Properties of region $P$]
    Given Assumption \ref{ass1}, region $P$ is positively invariant and globally attractive.
    \label{lem1}
\end{lem}
\begin{proof}
    For positive invariance, let $x(t)$ be a solution. If $x(t)$ enters $P$, both routes will be either in mode SF or UF. From \eqref{td}, we see that $x_l=C_l\Rightarrow\dot{x}_l\leq0,\ l=1,2$, ensuring that trajectories cannot escape $P$, thus guaranteeing that $x(t)\in P$, for all $t\geq\theta$. 
    
    For global attractivity, consider now $x\in Q$.
     Then from the definition of $Q$, at least one of the two routes is in mode SC or UC. If Route~$l$ is in mode SC or UC, we can write
        \begin{equation*}
           \dot{x}_l\leq-\frac{F_l}{L_l(B_l-C_l)}x_l+\frac{F_lC_l}{L_l(B_l-C_l)}.
        \end{equation*}
      This inequality implies the convergence to $P$.
    \end{proof}

    \begin{rem}[Traffic interpretation of the properties of region $P$]
    From a traffic perspective, the positive invariance of $P$ implies that, in our model, congestion cannot arise from a free-flow condition. Also, global attractiveness implies that congestion always vanishes over time. This property is due to the implicit assumption of infinite capacity at the destination node and unconstrained route outflows.
    \label{rem1}
\end{rem}
We shall prove that system \eqref{ch:dynselfish:os}-\eqref{dem} and \eqref{ch:dynselfish:rr} admits a globally asymptotically stable equilibrium, contained in region $P$. The proof relies on the fact that \eqref{ch:dynselfish:os}-\eqref{dem} and \eqref{ch:dynselfish:rr} admits a unique equilibrium and is monotone, in the following sense \cite{Hirsch}. 
\begin{defi}[Monotone system]
    A system $\dot{y}=h(y)$ with $h:\mathbb{R}^d\rightarrow\mathbb{R}^d$ is said to be \textit{monotone} if $y_0\leq \Tilde{y}_0$ 
    implies that $\varphi_t(y_0)\leq \varphi_t(\Tilde{y}_0)$, $\forall t\geq0$, where $\varphi_t(y^\circ)$ is the solution to $\dot{y}=h(y)$ with initial condition $y(0)=y^\circ$.
\end{defi}

\begin{prop}[Monotonicity]
   Given Assumptions \ref{ch:dynselfish:ass4} and \ref{ass1}, the system \eqref{ch:dynselfish:os}-\eqref{dem} and \eqref{ch:dynselfish:rr} is monotone.  
    \label{p1}
\end{prop}
\begin{proof}
    See Appendix \ref{app1}.
\end{proof}
\noindent Monotonicity imparts a high degree of structure to the system,
making it easier to establish its stability properties.
Before proceeding, we make the following assumption.
\begin{assum}[Demand upper-bound]
    The following conditions hold:
    \begin{equation}
        \Phi<v_lB_l,\quad l=1,2.
        \label{eqass6}
    \end{equation}
    \label{ass6}
\end{assum}
Although condition \eqref{eqass6} represents a formal constraint, it will not be restrictive in practical cases.

\begin{thm}[Global Asymptotic Stability]
    Given Assumptions \ref{ch:dynselfish:ass4}, \ref{ass1}, and \ref{ass6}, system \eqref{ch:dynselfish:os}-\eqref{dem} and \eqref{ch:dynselfish:rr} admits a globally asymptotically stable equilibrium $\overline{x}\in P$ .
    \label{ch:dynselfish:t2}
\end{thm}
\begin{proof}
    See Appendix \ref{app1}
\end{proof}
Theorem~\ref{ch:dynselfish:t2} characterizes the asymptotic behavior of \eqref{ch:dynselfish:os}-\eqref{dem} and \eqref{ch:dynselfish:rr}, ensuring that all solutions converge to a unique equilibrium $\overline{x} \in P$ in the free-flow regime. However, the convergence of the two-route sub-system does not fully determine the asymptotic behavior of the overall system, including buffer A. Specifically, depending on which sub-region within region P the equilibrium lies, the dynamics of A exhibit different behaviors. If the equilibrium falls within sub-region $\Omega^{\text{SF-SF}}$, the density of A converges to a finite value. Conversely, if the equilibrium is in sub-regions $\Omega^{\text{UF-SF}}$ or $\Omega^{\text{SF-UF}}$, the density of A becomes unbounded, leading the overall system to diverge. Therefore, it is crucial to establish the nature of the equilibrium, particularly whether it involves partial demand transfer, to gain insights into the traffic dynamics and assess the impact of the app's recommendations on traffic and network efficiency at steady-state.

\begin{rem}[Beyond piece-wise linear supply and demand functions]
Theorem~\ref{ch:dynselfish:t2} assumes that the supply and demand functions are piece-wise linear as per Equations \eqref{sup}-\eqref{dem}. However, this result can be extended to more general supply and demand functions. Specifically, it is possible to show (with minor changes to the proof) that the result still holds when the supply and demand functions for $l=1,2$ take the form
\begin{equation*}
s_l(x_l)=\min\{F_l, \Tilde{s}_l(x_l)\},\quad d_l(x_l)=\min\{F_l,\Tilde{d}_l(x_l)\}
\end{equation*}
where $\Tilde{s}_l:[0,B_l]\rightarrow\mathbb{R}_{>0}$ is a strictly decreasing $C^1$ function such that $\Tilde{s}_l(C_l)=F_l,\ \Tilde{s}_l(\overline{x}_l)=0$, and $\Tilde{d}_l:[0,B_l]\rightarrow\mathbb{R}_{>0}$ is a strictly increasing and $\Tilde{d}_l(0)=0,\ \Tilde{d}_l(C_l)=F_l$. In this case, Assumption \ref{ass6} should be replaced by $\Phi<\Tilde{d}_l^{-1}(B_l)$. 
\end{rem}

\section{Equilibrium efficiency for logit routing}\label{sec4}
This section is dedicated to the analysis of the traffic implications of the model, focusing on the study of its unique equilibrium: the latter being globally asymptotically stable, its properties fully describe the steady-state of the system. 
Studying these properties in full generality is made hard by the lack of an analytic characterization of the equilibrium.
For this reason, from now on we shall study system \eqref{td} with logit routing ratios~\eqref{lrr} and affine travel times, i.e.,
\begin{equation}
    \tau_l(x_l)=a_l\frac{x_l}{B_l}+\frac{L_l}{v_l},\quad l=1,2.
    \label{aft}
\end{equation}
Affine travel time functions are largely used in the traffic literature, especially when considering a free-flow regime on a route \cite{wollenstein2,wollenstein3,friesz}.
These two assumptions allow for the detailed analysis of $\overline{x}$ in the {\em high compliance} and {\em low compliance} regimes.
Under high user compliance, i.e., when $\eta\rightarrow0$, we will prove that the equilibrium converges to the Wardrop equilibrium of an underlying non-atomic routing game: the properties of the Wardrop equilibrium can be extended by continuity to the equilibrium\footnote{The fact that the equilibria of a logit-based dynamics converge to the Wardrop equilibria of an associated game has already been exploited in the literature \cite{sandholm}, including for similar traffic models that did not account for route's capacity saturation~\cite{cianfanelli2022}.}.
For low compliance, we will show that a linearization of \eqref{td} equipped with \eqref{lrr} provides a suitable approximation.

The main criterion that will be used to evaluate the efficiency of the equilibrium is to establish whether it features partial demand transfer or not. In order to disregard trivial cases in which unsatisfied demand arises independently of routing recommendations, we make the following assumption:
\begin{assum}[Fixed routing ratios]    \label{ass5}
    The following condition holds:
    \begin{equation}  
        F_l>(1-\alpha)\Phi r_l^0,\quad l=1,2,\quad \forall\alpha.
    \end{equation}
\end{assum}

For fully transferring equilibria, we will also consider the total travel time as a second criterion. In particular, we will evaluate the Price of Anarchy,  which corresponds to the ratio between the total travel time attained at equilibrium and the minimum total travel time attained by the social optimum:
\begin{equation}
    \poa(x^*)=\frac{\sum_{l=1,2}\Phi R_l(x^*)\tau_l(x_l^*)}{\sum_{l=1,2}\Phi R_l^O\tau_l(x_l^O)},
    \label{total}
\end{equation}
where $(R_l^O,x^O)$ is such that $R_l^O=\frac{v_lx_l^O}{\Phi}$, $\Phi R_l\leq F_l,\ l=1,2$, and minimizes the total travel time. We will see that this second criterion leads us to conclusions similar to those in \cite{thai,ibp,th:cabannes}.

\subsection{High drivers' compliance}
As previously mentioned, when user compliance is very high, the properties of the equilibrium of \eqref{ch:dynselfish:os}-\eqref{dem} and \eqref{ch:dynselfish:rr} can be inferred from those of the Wardrop equilibrium of an underlying non-atomic routing game. We will first introduce the game and deduce some of its properties, then extend them by continuity to $\overline{x}$. 

\subsubsection{Underlying routing game}
Consider the two-route sub-network introduced in Section~\ref{sec2} (see Figure~\ref{grp}).  Assume that the system is subject to a demand of users $\Phi$, aiming at crossing it to reach the destination node. Suppose that $\Phi$ consists of both informed users and uninformed users. Informed users represent a fraction $\alpha\in[0,1]$, whereas the remaining part is represented by uninformed users. As for system \eqref{ch:dynselfish:os}, uninformed users split on the two routes according to prior beliefs $r^0=(r_1^0,r_2^0)$, while informed users choose their route to minimize their travel time, according to the travel time functions \eqref{aft}. The splitting of the exogenous flow between the two routes is captured by the routing ratios $R_1$, $R_2$, so that the flow sent on Route 1 and Route 2 is $\Phi R_1$ and $\Phi R_2$, respectively. Observe that the routing ratios of this game are such that 
\begin{equation}
    R_1\in[(1-\alpha)r_1^0,(1-\alpha)r_1^0+\alpha],\quad R_2=1-R_1,
\end{equation}
due to the fact that the splitting of uninformed users is fixed.

The goal of non-atomic routing games is to identify, among the steady-state flows, those that are likely to emerge from user interaction. Therefore, we will also focus exclusively on steady-state flows, which, according to \eqref{ch:dynselfish:os}, correspond to flows that satisfy the equilibrium condition
\begin{equation*}
    \min\{\Phi R_l,s_l(x_l)\}=d_l(x_l),\qquad l=1,2.
\end{equation*}
Notice that the above condition implies two important facts. First, no steady-state flow can be characterized by a congested route, i.e., all steady-sate flows of the game are included in region $P$, which is consistent with the behavior of system \eqref{ch:dynselfish:os}. Second, route flows induce the following densities on the two routes:
\begin{equation}
    x_l^R=\begin{cases}
        \frac{\Phi R_l}{v_l}, & \Phi R_l<F_l\\
        C_l, & \Phi R_l\geq F_l
    \end{cases},\qquad l=1,2.
    \label{ch:dynselfish:xW}
\end{equation}
\begin{defi}[Underlying routing game (URG($\alpha$))]
    The underlying routing game URG($\alpha$) of system \eqref{ch:dynselfish:os} is the quadruplet $(\mathcal{G},\Phi,\tau,\alpha)$, where $\mathcal{G}$ is the network which the game is defined on, $\Phi$ is the demand for the network, $\tau$ is the set of link travel times functions, and $\alpha$ is the penetration rate of informed users.
\end{defi}
\begin{defi}[Wardrop equilibrium (WE)]
    A pair $(R^W(\alpha),x^W(\alpha))$ of the URG($\alpha$) is a \emph{Wardrop equilibrium (WE($\alpha$))} if and only if 
    \begin{equation}
        R_l^W(\alpha)>(1-\alpha)r_l^0\ \Rightarrow\ \tau_l(x_l^W(\alpha))\leq \tau_k(x_k^W(\alpha)),\quad k\neq l,\ l=1,2.
    \end{equation}
\end{defi}

In this case, the WE is expressed as a function of $\alpha$, to emphasize the influence of the penetration rate on its shape.
We next characterize the Wardrop equilibrium $\mathrm{WE}(\alpha)$ of the NRG. Let us define the following quantities to ease the notation:

Let us define the following quantities to ease the notation:
\begin{equation*}
   c_l:=\frac{a_l}{v_lB_l},\quad b_l:=\frac{L_l}{v_l},\quad l=1,2.
\end{equation*}

For convenience, we make the following assumption on the route travel times.
\begin{assum}[Route labeling]
    Route $1$ has the shortest travel time route for $\alpha=0$, i.e., $b_1+c_1\Phi r_1^0\leq b_2+c_2\Phi r_2^0$.
    \label{me1}
\end{assum}
\noindent For ease in stating our results, we also define the following quantities:
\begin{equation}
    \overline{\Phi}:=F_1\left(1+\frac{c_1}{c_2}\right)-\frac{b_2-b_1}{c_2},
    \label{of}
\end{equation}
\begin{equation}
\alpha^\mathrm{M}:=\frac{1}{\Phi r_2^0}\frac{c_2\Phi r_2^0-c_1\Phi r_1^0+b_2-b_1}{c_1+c_2},
    \label{sa}
\end{equation}
\begin{equation}
\alpha^\mathrm{U}:=\frac{F_1-\Phi r_1^0}{\Phi r_2^0},
    \label{ba}
\end{equation}
\begin{equation}
\alpha^\mathrm{UM}:=1-\frac{1}{\Phi r_2^0}\left(\frac{c_1}{c_2}F_1-\frac{b_2-b_1}{c_2}\right).
\end{equation}

 \begin{lem}[Unsatisfied demand, Wardrop equilibrium]
    Suppose Assumptions \ref{ass1}, \ref{ass6}, \ref{ass5} and \ref{me1} holds. 
    Then, the underlying routing game admits a unique Wardrop equilibrium $(R^W(\alpha),x^W(\alpha))$ and the following characterization holds:
    \begin{itemize}
        \item If $\Phi\leq\overline{\Phi}$, then no route is affected by unsatisfied demand. Moreover,
        \begin{itemize}
            \item if $\alpha\leq\alpha^\mathrm{M}$, then
            \begin{equation}
        \begin{aligned}
            &R_1^W(\alpha)=\alpha+(1-\alpha) r_1^0,\\
            &R_2^W(\alpha)=(1-\alpha) r_2^0,\\
            &x^W(\alpha)=\left(\frac{\Phi R_1^W(\alpha)}{v_1},\frac{\Phi R_2^W(\alpha)}{v_2}\right);
        \end{aligned}
        \label{allin1}
        \end{equation}
        \item if $\alpha>\alpha^\mathrm{M}$, then
        \begin{equation}
        \begin{aligned}
            &R_l^W(\alpha)=
                     \frac{c_{k}\Phi+b_{k}-b_l}{c_l+c_{k}},\quad k\neq l,\ l=1,2,\\&x^W(\alpha)=\left(\frac{\Phi R_1^W(\alpha)}{v_1},\frac{\Phi R_2^W(\alpha)}{v_2}\right).
        \end{aligned}
             \label{awe}
        \end{equation}
        \end{itemize}
        \item If $\Phi>\overline{\Phi}$, then Route $1$ will be affected by unsatisfied demand for $\alpha>\alpha^\mathrm{U}$. Moreover:
        \begin{itemize}
            \item if $\alpha\leq\alpha^\mathrm{U}$, then $(R^W(\alpha),x^W(\alpha))$ is as in \eqref{allin1};
       \item if $\alpha\leq\alpha^\mathrm{U}$, then $R^W(\alpha)$ is as in \eqref{allin1} and 
       \begin{equation}
           x^W(\alpha)=\left(x_1^c,\frac{(1-\alpha)\Phi r_2^0}{v_2}\right);
       \end{equation}
        \item if $\alpha>\alpha^\mathrm{UM}$, then
        \begin{equation}
            \begin{aligned}
            &R_1^W(\alpha)=\Phi-\frac{c_1}{c_2}F_1+\frac{b_2-b_1}{c_2},\\ &R_2^W(\alpha)=\frac{c_1}{c_2}F_1-\frac{b_2-b_1}{c_2},\\& x^W(\alpha)=\left(x_1^c,\frac{(1-\alpha)\Phi r_2^0}{v_2}\right);
            \end{aligned}
            \label{sawe}
            \end{equation}
        \end{itemize}
    \end{itemize}
    \label{lem2}
 \end{lem}
 \begin{proof}
     From the definition of Wardrop equilibrium, informed users distribute only on shortest travel time routes. From Assumption \ref{me1}, it is clear that for $\alpha$ small enough, $\alpha\Phi$ will distribute entirely on Route $1$. As $\alpha$ increases, the demand routed toward Route $1$ increases as well, until either travel times equalize or Route $1$ gets saturated.  $\alpha^M$ in expression \eqref{sa} is the maximum value of $\alpha$ such that, by allocating all informed users on Route $1$, Route $1$ still is the shortest travel time route, and can be calculated by solving $\tau_1((1-\alpha)\Phi r_1^0+\alpha\Phi)=\tau_2((1-\alpha)\Phi r_2^0)$.
     $\alpha^U$ in expression \eqref{ba}, instead, is the maximum value of $\alpha$ such that, by allocating all informed users on Route $1$, the capacity on Route $1$ is not exceeded and can be calculated by solving $(1-\alpha)\Phi r_l^0+\alpha\Phi=F_l$. One can verify that $\Phi\leq\overline{\Phi}\iff \alpha^\mathrm{M}\leq\alpha^\mathrm{U}$.
     
     This leads to consider two cases, depending on the value of $\Phi$. If $\Phi\leq\overline{\Phi}$, then  by increasing $\alpha$, travel times equalize before Route $1$ gets congested, since $\alpha^\mathrm{M}\geq\alpha^\mathrm{U}$. Once travel times are equal, $R^W(\alpha)$ takes the form in \eqref{awe}, which no longer depends on $\alpha$. Hence, increasing $\alpha$ will affect no more the shape of $R^W(\alpha)$, as the additional app-informed users will distribute on the two routes so as to keep travel times even. Expression \eqref{awe} can be retrieved by imposing the two routes' travel times to be equal. On the contrary, if $\Phi>\overline{\Phi}$, then Route $1$ gets congested before travel times equalize. Nevertheless, analogously to the previous case, the informed demand keeps selecting Route $1$, which is still the shortest travel time route, until travel times equalize, which now happens at $\alpha^\mathrm{UM}$. Again, one can verify that $\Phi\geq\overline{\Phi}\iff \alpha^\mathrm{U}\leq\alpha^\mathrm{UM}$. After travel times even out, $R^W(\alpha)$ takes the form in \eqref{sawe} and further increase do not affect the demand distribution anymore. Expression \eqref{sawe} can be retrieved by imposing the two routes' travel times to be equal, accounting for the fact that Route $1$ is saturated.

     The uniqueness of $R^W(\alpha)$ follows from the fact that the above cases are exhaustive and mutually exclusive. To conclude, the expressions of $x^W(\alpha)$ easily follow from \eqref{ch:dynselfish:xW}.
 \end{proof}

 \subsubsection{Convergence to the Wardrop equilibrium and its implications}
We now prove how $x^*$ converges to $x^W(\alpha)$.
\begin{prop}[Equilibrium approximation]
    Let Assumptions 
    \ref{ass1}, \ref{ass6}, \ref{ass5} hold. The unique equilibrium $x^*$ of \eqref{ch:dynselfish:os} equipped with logit routing ratios \eqref{lrr} converges to $x^W(\alpha)$, as $\eta\rightarrow0$.
\end{prop}
\begin{proof}
From the proof of Theorem \ref{ch:dynselfish:t2}, $x^*$ corresponds to the fixed point of the map
\begin{equation*}
    G_l(x_l,\eta)=\frac{\min\{\Phi R_l(\tau(x_l)),F_l\}}{v_l},\quad l=1,2,
\end{equation*}
with $R_l(\tau(x))$ as in \eqref{lrr}. Consider a sequence $(\eta_n)_{n\in\mathbb{N}}|\eta_n\rightarrow0$. Let ${x^*}^{(n)}$ the unique equilibrium associated with the corresponding instance of \eqref{ch:dynselfish:os}. Since ${x^*}^{(n)}\in P,\ \forall n$, and $P$ is compact, the sequence $\{{x^*}^{(n)}\}_{n\in\mathbb{N}}$ is bounded. From compactness, every sequence admits a converging sub-sequence. Let $E\subseteq P$ be the set of accumulation points of all converging sequences of equilibria of \eqref{ch:dynselfish:os}. Pick $e\in E$ and the corresponding sequence $\{{x^*}^{(n_k)}\}_{k\in\mathbb{N}}$. Assume that one of the two routes is sub-optimal at $e$, i.e., $\exists l\,|\,\Tilde{\tau}_l(e_l)>\Tilde{\tau}_{l'}(e_{l'})$, for some $l$. Then
\begin{equation}
    e_l=\lim_{\eta\rightarrow0}{x^*}^{(n_k)}=\lim_{\eta\rightarrow0} G_l({x^*}^{(n_k)},\eta)=\frac{(1-\alpha)\Phi r_l^0}{v_l}.
\end{equation} 
Hence, $e$ corresponds to $x^W(\alpha)$, since none of the informed demand $\alpha\Phi$ is allocated on the sub-optimal route $l$. Since $(R^W(\alpha),x^W(\alpha))$ is the unique WE of the underlying routing game, it follows that all sequences converge to $e=x^W(\alpha)$. Again, compactness ensures that all the sub-sequences of $\{{x^*}^{(n_k)}\}_{n\in\mathbb{N}}$ admit a sub-sequence converging to $x^W(\alpha)$. This is equivalent to say that $\{{x^*}^{(n_k)}\}_{n\in\mathbb{N}}$ converges to $x^W(\alpha)$, as well.
\end{proof}


Since $x^*$ converges to $x^W(\alpha)$, by continuity we can extend the properties of the latter to the former.
\begin{corol}[Partial demand transfer for high compliance]
Let Assumptions 
\ref{ass1}, \ref{ass6}, \ref{ass5} and \ref{me1} hold. If $\Phi>\overline{\Phi}$ and $\alpha>\alpha^\mathrm{U}$, then, for small enough $\eta$, the equilibrium $x^*$ of \eqref{ch:dynselfish:os}-\eqref{dem} and \eqref{ch:dynselfish:rr} is affected by partial demand transfer (Route~1 is saturated) and $x_A\rightarrow+\infty,\ t\rightarrow+\infty$.
\label{ch:dynselfish:cor1}
\end{corol}
Lemma \ref{lem2} and Corollary \ref{ch:dynselfish:cor1} highlight the significance of the parameters $\alpha$ and $\Phi$ to the problem, showing that a higher penetration rate increases the system's susceptibility to partial demand transfer. Moreover, a higher traffic demand heightens the system's sensitivity to the penetration rate, lowering the threshold $\alpha^\mathrm{U}$ beyond which partial demand transfer occurs.

We now investigate what impact the penetration rate has on the Price of Anarchy, under Assumptions \ref{ass5} and \ref{me1}. We will first perform the analysis on $(R^W(\alpha),x^W(\alpha))$ and then extend it by continuity to $x^*$ when \eqref{ch:dynselfish:os}-\eqref{dem} and \eqref{ch:dynselfish:rr} is equipped with logit routing ratios and affine travel times, in the limit of high compliance. Since we consider $\poa$ meaningful only when there is full demand transfer, we assume $\Phi\leq\overline{\Phi}$. 
With abuse of notation, we will write
\begin{equation*}
    \poa(\alpha)=\frac{\sum_{l=1,2}\Phi R_l^W(\alpha)\tau_l(x_l^W(\alpha))}{\sum_{l=1,2}\Phi R_l^O\tau_l(x_l^O)}.
\end{equation*}
\begin{prop}[Price of anarchy in URG($\alpha$)]\label{p3} Suppose that Assumptions \ref{ass5} and \ref{me1} holds and that $\Phi\leq\overline{\Phi}$.
Then, $\poa(\alpha)$ is strictly convex in $\alpha$ in $[0,\alpha^\mathrm{M}]$ and constant for $\alpha>\alpha^\mathrm{M}$. Moreover, let 
\begin{equation}
    \alpha^{\text{opt}}:=\frac{1}{\Phi r_2^0}\frac{2c_2\Phi r_2^0-2c_1\Phi r_1^0+b_2-b_1}{2(c_1+c_2)}.
\end{equation}
Then,
\begin{itemize}
    \item if $\alpha^{\text{opt}}<0$, then  $\poa(\alpha)$ is increasing in $[0,\alpha^\mathrm{M}]$ and constant in $[\alpha^\mathrm{M},1]$;
    \item if $\alpha^{\text{opt}}\in[0,\alpha^\mathrm{M}]$, then $\poa(\alpha)$ attains its minimum at $\alpha^{\text{opt}}$;
    \item if $\alpha^{\text{opt}}> \alpha^\mathrm{M}$, then $\poa(\alpha)$ is decreasing in $[0,\alpha^\mathrm{M}]$ and constant in $[\alpha^\mathrm{M},1]$.
\end{itemize}
\end{prop}
\begin{proof}
Strict convexity in $[0,\alpha^\mathrm{M}]$ can be checked through a second order condition. 
The fact that $\poa(\alpha)$ is constant for $\alpha>\alpha^\mathrm{M}$ follows from $R^W(\alpha)$ is constant for $\alpha>\alpha^\mathrm{M}$. The expression of $\alpha^{\text{opt}}$ can be retrieved by solving $\partial_\alpha \poa(\alpha)=0$.
\end{proof}

One can verify that $\alpha^\text{opt}<\alpha^\mathrm{M}$ if and only if $b_1<b_2$, which means that Route $1$ is faster than Route $2$ when both are empty. 
Therefore, in this case an excessive penetration rate, notably $\alpha>\alpha^\text{opt}$, leads to an increased number of users favoring the shortest travel time route. Consequently, this elevates the average travel time for users, thereby reducing the efficiency of the network.

Considering Proposition 2 and that the PoA is a continuous function of the density, the above characterization of the average travel time at Wardrop equilibrium as a function of the penetration rate constitutes a good approximation of the average travel time at $x^*$  equipped with logit routing ratios \eqref{lrr}, for sufficiently small $\eta$. 

\subsection{Low drivers' compliance}
So far we have analyzed $x^*$ under logit routing ratios \eqref{lrr}, assuming high users' compliance. The case of low compliance can be addressed by studying a suitable linearization of the system. Its analysis yields results that are qualitatively consistent with those for high compliance.

Assume that $|\tau_1-\tau_2|/\eta\rightarrow0$, i.e., the argument of the exponential in \eqref{lrr} approaches zero. Then, \eqref{lrr} admits the following first-order approximation:
\begin{equation}
    R_l(\tau)=r_l^0+\frac{\alpha r_l^0r_j^0}{\eta}\left(\tau_j-\tau_l\right),\quad i\neq j,\ l=1,2.
    \label{foa}
\end{equation}
For Eq.~\eqref{foa} to be a valid approximation of \eqref{lrr}, i.e., to guarantee that \eqref{foa} satisfies to $0\leq R_l(\tau)\leq1,\ l=1,2$, it is necessary that 
\begin{equation}
      \frac{1}{\eta}\leq\frac{1}{\alpha\Delta\max_\ell r_\ell^0}
      \label{dom}
\end{equation}
where $\Delta:=\max_{\tau\in\mathbb{R}_{>0}^2}|\tau_1-\tau_2|$. 

Since \eqref{foa} is a monotone routing ratio, Theorem~\ref{ch:dynselfish:t2} holds and \eqref{ch:dynselfish:os}-\eqref{dem} and \eqref{ch:dynselfish:rr} admits a globally asymptotically stable equilibrium, which can now also be calculated explicitly. So, we can proceed to analyse demand transfer and network performance.


\begin{prop}[Partial demand transfer for low compliance]\label{ch:dynselfish:p4}
    Let us assume that Assumptions~\ref{ass1}, \ref{ass6}, \ref{ass5} and condition \eqref{dom} hold. 
    If $\Phi>\overline{\Phi}$ and
    \begin{equation}
    \alpha>\Tilde{\alpha}^U:=\frac{\eta}{r_1^0(c_2\Phi+b_2-b_1-F_1(c_1+c_2))}\alpha^\mathrm{U},
    \label{lcalpha}
    \end{equation}
    then the equilibrium $x^*$ of \eqref{ch:dynselfish:os}-\eqref{dem} and \eqref{ch:dynselfish:rr} equipped with routing ratios as in \eqref{foa} presents is affected by partial demand transfer (Route~1 is saturated) and $x_A(t)\rightarrow+\infty,\ t\rightarrow+\infty$.
\end{prop}
\begin{proof}
    In this case, 
    \begin{equation*}
        x^*_l=\frac{\eta\Phi r_l^0+\alpha\Phi r_l^0r_k^0(c_k\Phi+b_k-b_l)}{v_l(\eta+\alpha\Phi r_l^0r_k^0(c_l+c_k))},\quad l=1,2.
    \end{equation*}
    Unsatisfied demand emerges on Route $1$ when $v_1x^*_1>F_1$. By plugging into the latter condition the above expression of $x^*_1$, after rearranging some terms one finds the following equivalent condition:
    \begin{equation*}
        \alpha\Phi r_1^0r_2^0(c_2\Phi+b_2-b_1-F_1(c_1+c_2))>\eta(F_1-\Phi r_1^0).
    \end{equation*}
    This condition is met if and only if $\Phi>\overline{\Phi}$ and \eqref{lcalpha} holds.
\end{proof}
Proposition \ref{ch:dynselfish:p4}, akin to Lemma \ref{lem2} and Corollary \ref{ch:dynselfish:cor1}, demonstrates that a higher traffic demand increases the system's sensitivity to the penetration rate. In fact, $\Tilde{\alpha}^\mathrm{U}$ is decreasing in $\Phi$.

\begin{prop}[Price of Anarchy for low compliance]\label{ch:dynselfish:p5}If Assumptions~\ref{ass1}, \ref{ass6}, \ref{ass5}, \ref{me1} and condition~\eqref{dom} hold, then the average travel time of equilibrium $x^*$ under routing ratios \eqref{foa}, $\poa(\alpha)$, is convex in $\alpha$. Moreover, let \begin{equation}
        \Tilde{\alpha}^\text{opt}:=\frac{2\eta}{r^0_1(b_2-b_1)}\alpha^\text{opt}.
        \label{lopt}
    \end{equation}
The following three cases are possible:    
    \begin{itemize}
    \item if $\Tilde{\alpha}^{\text{opt}}<0$, then $\poa(\alpha)$ is increasing in $\alpha$;
    \item if $\Tilde{\alpha}^{\text{opt}}\in[0,1]$, then $\poa(\alpha)$ attains its minimum at $\Tilde{\alpha}^{\text{opt}}$;
    \item if $\Tilde{\alpha}^{\text{opt}}>1$, then $\poa(\alpha)$ is decreasing in $[0,1]$.
\end{itemize}
\end{prop}
We omit its proof, since it is analogous to that of Proposition \ref{p3}. 
One can notice the similarity in the conditions provided in Propositions~\ref{p3} and~\ref{ch:dynselfish:p5}. Another interesting aspect is that $\Tilde{\alpha}^\text{opt}$ is proportional to $\eta$, which means that for lower levels of compliance, it takes a higher penetration rate to attain the optimum.

\section{Numerical simulations}\label{sec5}
This section presents numerical simulations corroborating the theoretical findings presented in Section \ref{sec4}. In Section \ref{macrosim}, we introduce macroscopic simulations, while Section \ref{microsim} focuses on microscopic simulations conducted using the Aimsun software. 

\subsection{Macroscopic simulations}\label{macrosim}

\begin{figure*}
    \centering
    \includegraphics[width=.32\textwidth]{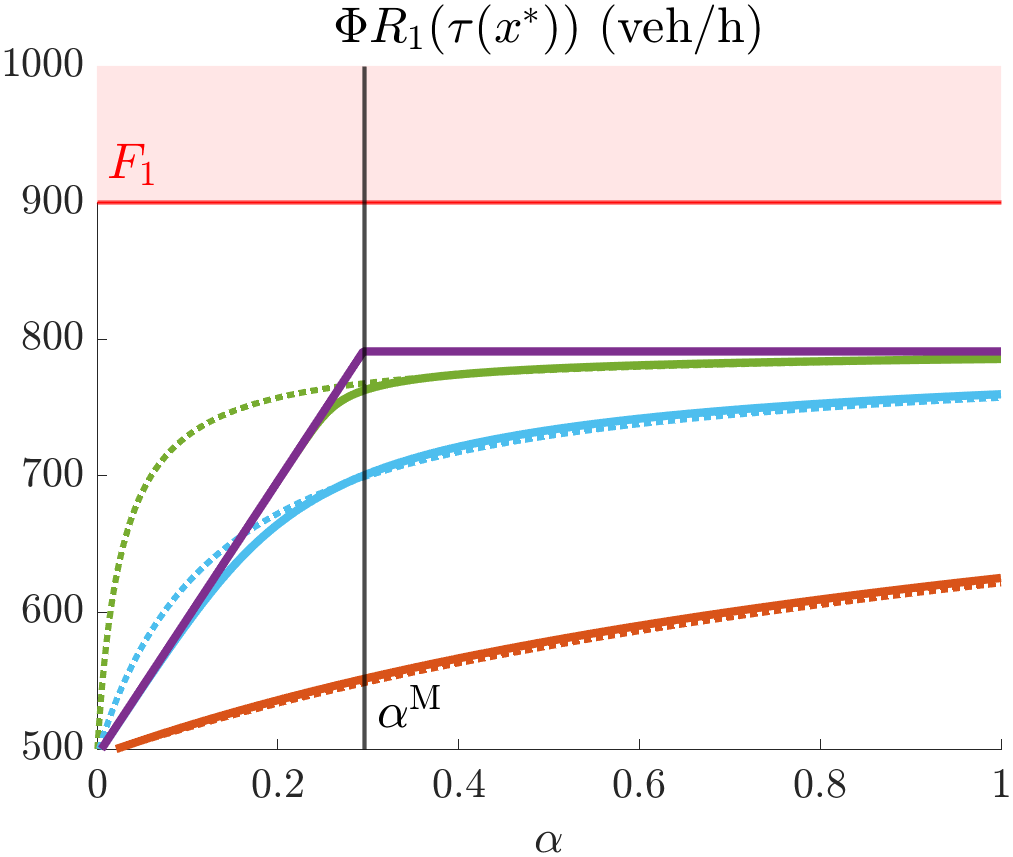}\,
    \includegraphics[width=.32\textwidth]{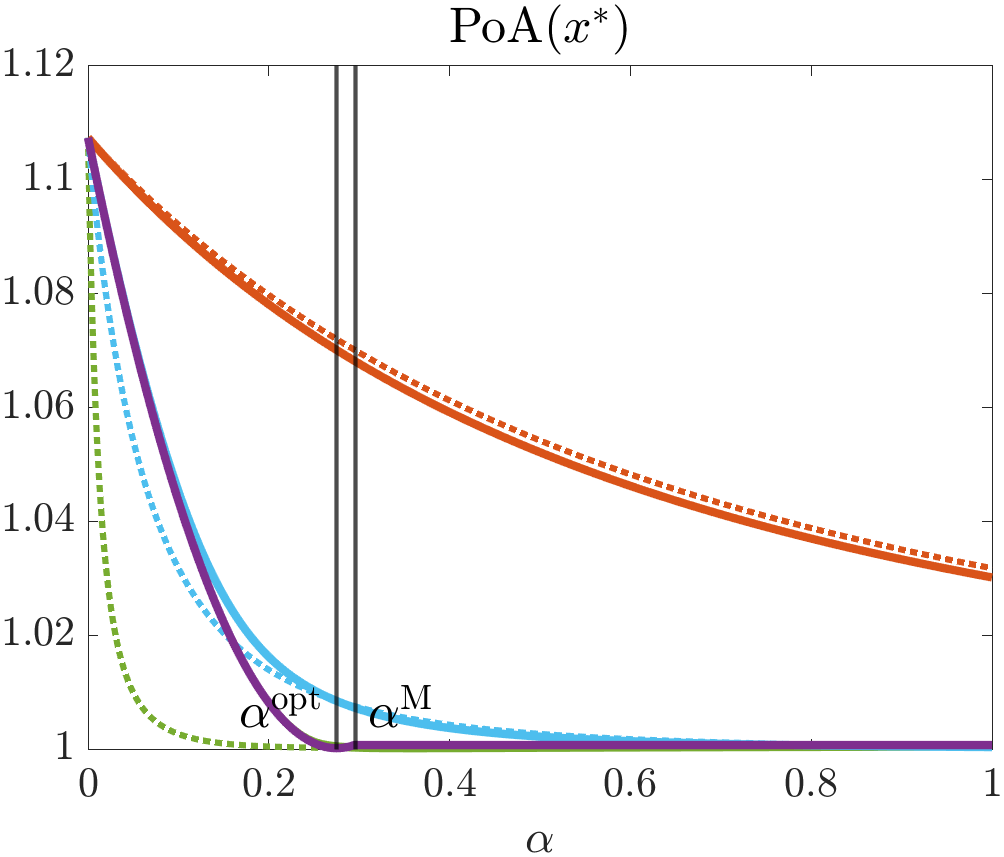}\, \includegraphics[width=.32\textwidth]{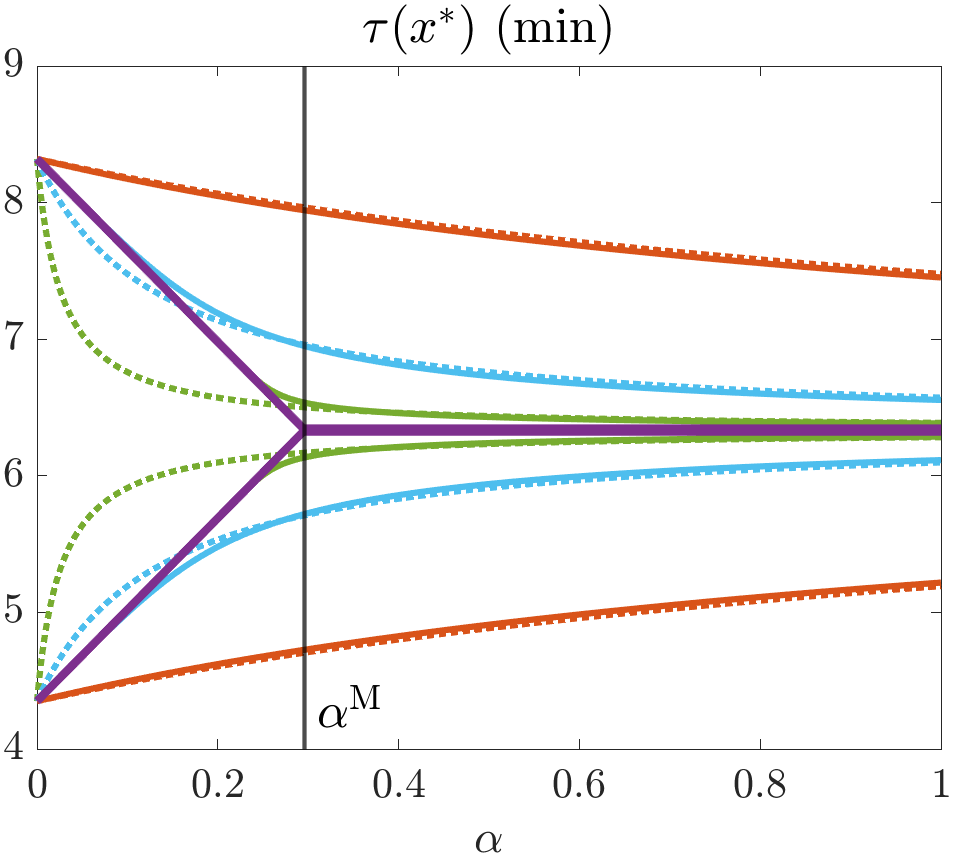} 
    \caption{Simulations for $\Phi=1500$. The three plots respectively show the demand directed toward Route $1$ (left), average travel time $T$ (middle) and route travel times (right) at equilibrium, as functions of penetration rate $\alpha$. The diamond-marked lines are the logit routing ratios, while the dashed lines correspond to the linearized model. 
    We draw the curves for $1/\eta=10$ in \textcolor{BurntOrange}{orange}, for $1/\eta=100$ in \textcolor{ProcessBlue}{light-blue},  and for $1/\eta=500$ in \textcolor{LimeGreen}{green}. The limit Wardrop equilibrium $\mathrm{WE}(\alpha)$ is drawn as \textcolor{Plum}{solid violet} lines. In the left-most plot, the area highlighted in \textcolor{red}{red} identifies the cases in which the demand toward Route $1$ is unsatisfied. In the right-most plot, the increasing travel time refers to Route $1$, the decreasing one to Route $2$.}
    \label{phi3}
\end{figure*}
\begin{figure*}
    \centering
    \includegraphics[width=.32\textwidth]{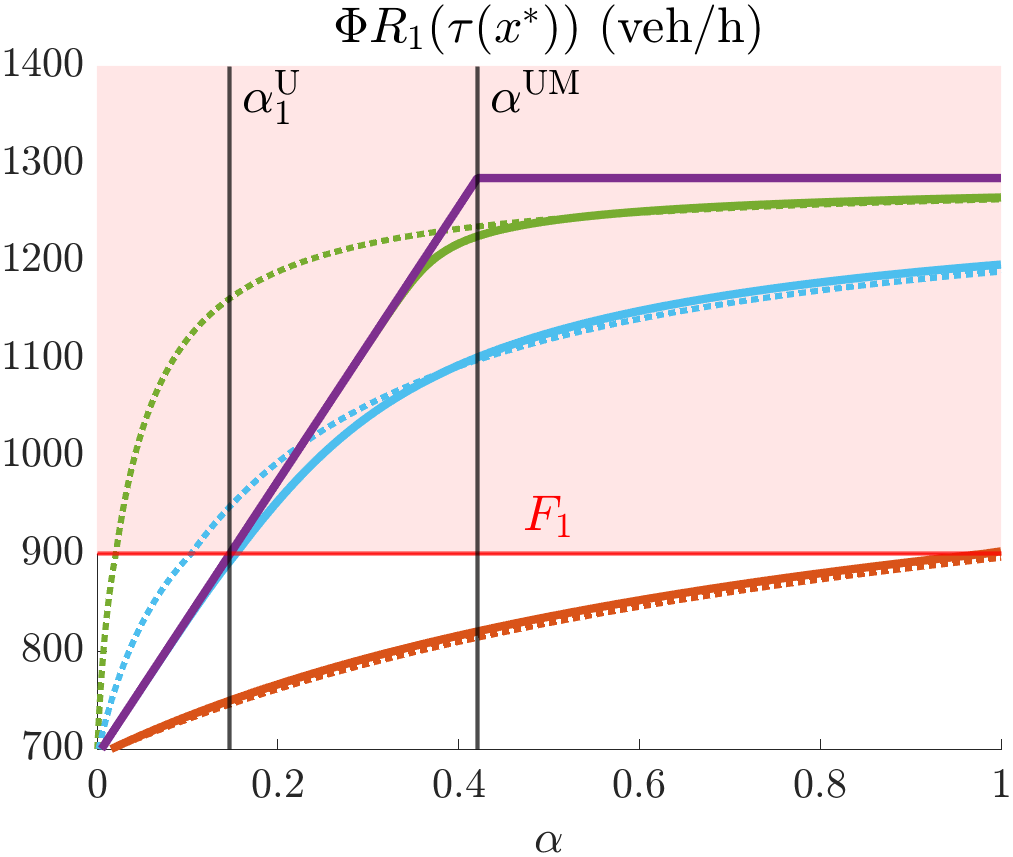}\ \includegraphics[width=.32\textwidth]{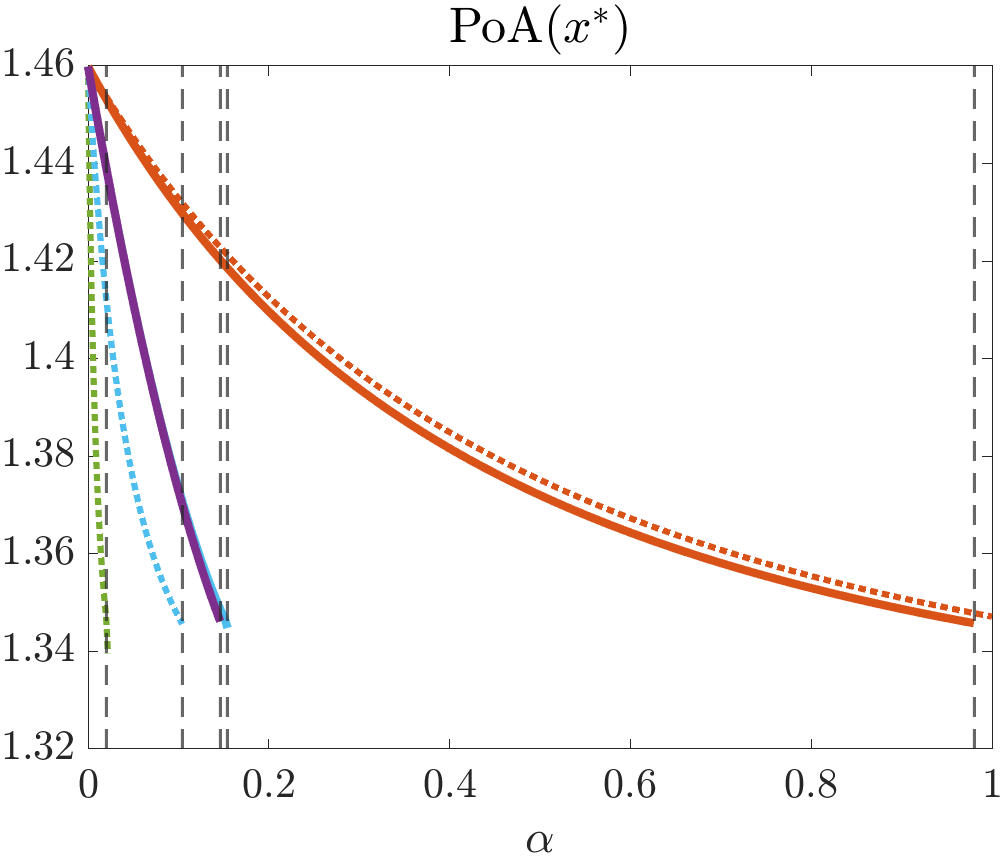}\ \includegraphics[width=.32\textwidth]{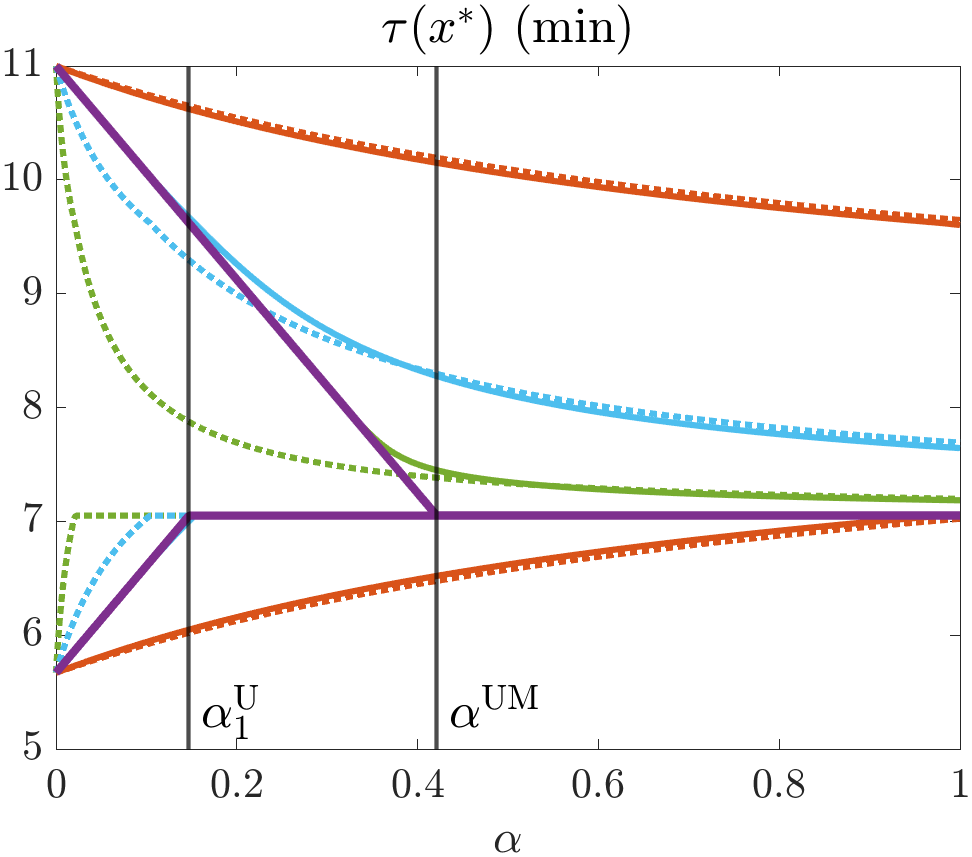}
    \caption{Simulations for $\Phi=2100$. The three plots respectively show the demand directed toward Route $1$ (left), average travel time $T$ (middle) and route travel times (right) at equilibrium, as functions of the penetration rate $\alpha$. The diamond-marked lines are the logit routing ratios, while the dashed lines correspond to the linearized model. 
    We draw the curves for $1/\eta=10$ in \textcolor{BurntOrange}{orange}, for $1/\eta=100$ in \textcolor{ProcessBlue}{light-blue}, and for $1/\eta=500$ in \textcolor{LimeGreen}{green}. The limit Wardrop equilibrium $\mathrm{WE}(\alpha)$ is drawn as \textcolor{Plum}{solid violet} lines. In the left-most plot, the area highlighted in \textcolor{red}{red} identifies the cases in which the demand toward Route $1$ is unsatisfied. In the middle plot, the lines are truncated at the value of $\alpha$ at which unsatisfied demand emerges. In the right-most plot, the increasing travel time refers to Route $1$, the decreasing one to Route $2$.} 
    \label{phi4}
\end{figure*}

In this section, we propose numerical simulations showcasing the qualitative behavior of system \eqref{ch:dynselfish:os}-\eqref{ch:dynselfish:rr}. 
We consider the following set of parameters:
\begin{equation*}
    \begin{aligned}
    &F_1=900 \text{ veh/h},\ F_2=1800 \text{ veh/h}\\
    &v_1=50 \text{ km/h},\ v_2=50 \text{ km/h}\\
    &B_1=90 \text{ veh/km},\ B_2=180 \text{ veh/km},\\
    &L_1=0.875\text{ km},\ L_2=1.35\text{ km},\\
    &a_1=0.5 \text{ h},\ a_2=1 \text{ h},\  r^0=(0.33,0.67). 
    \end{aligned}
\end{equation*}
 These parameters have been chosen to represent realistic travel conditions within an urban network, where Route~$1$ corresponds to an itinerary on a one-lane road, whereas Route~$2$ emulates an itinerary on a two-lane road. The choice of $r^0$ is motivated by the observation that drivers prefer major roads with higher capacity when no information about traffic state is available. 

In Figures~\ref{phi3} and \ref{phi4}, we provide simulations for two realistic values of demand ($\Phi=1500$ and $\Phi=2100$) and for three values of compliance ($1/\eta=10$, $1/\eta=100$ and $1/\eta=500$), which cover both high compliance and low compliance cases. First of all, we notice that, as expected, the approximation provided by $\mathrm{WE}(\alpha)$ better suits the logit model for high values of compliance, whereas the linearized model better captures low compliance scenarios. Concerning the qualitative behavior of the equilibrium point $x^*$ as a function of $\alpha$, the experiments reflect the theoretical results provided in Section~\ref{sec4}. To see this, let us comment upon Figures~\ref{phi3} and~\ref{phi4} more in detail. Notice that $\overline{\Phi}_1\approx1715\text{ veh/h},\ \overline{\Phi}_2=3685\text{ veh/h}$ and that Assumption \ref{me1} is satisfied for both values of demand.
%
%
The case of $\Phi=1500$ is illustrated in Figure \ref{phi3}: 
in this case, $\alpha^\mathrm{opt}=0.275,\ \alpha_1^\mathrm{M}\approx0.296$. Consistently with the theoretical results, unsatisfied demand does not arise on any route, for any of the parameter configurations.  For high compliance, the PoA reaches its minimum for low values of penetration rates. For low compliance, instead, the PoA is lsowly decreasing in $\alpha$. 
%
%
The case of $\Phi=2100$ is shown in Figure \ref{phi4}: in this case, $\Phi>\overline{\Phi}_1$ and partial demand transfer emerges on Route~1 for $1/\eta=100$ and $1/\eta=500$ for low values of $\alpha$ (around 0.15), and for $1/\eta=10$ for values of $\alpha$ close to 1. In both cases, the emergence of partial demand transfer translates into the density of the buffer growing unbounded over time (not shown in the plots). 

A relevant fact emerging from these simulations is that greater drivers' compliance can result in a decrease in the efficiency of the equilibrium, especially for high penetration rates. 
Since both high penetration rate and high compliance imply more informed users, these results align with the evidence in the literature about the information paradoxes in traffic networks \cite{ibp,bryce2024,wu}.

\subsection{Microscopic simulations}\label{microsim}
In this section, we propose agent-based simulations performed with the Aimsun micro-simulator \cite{aimsun}. 
\begin{figure*}[!t]
    \centering
    \includegraphics[width=.4\textwidth]{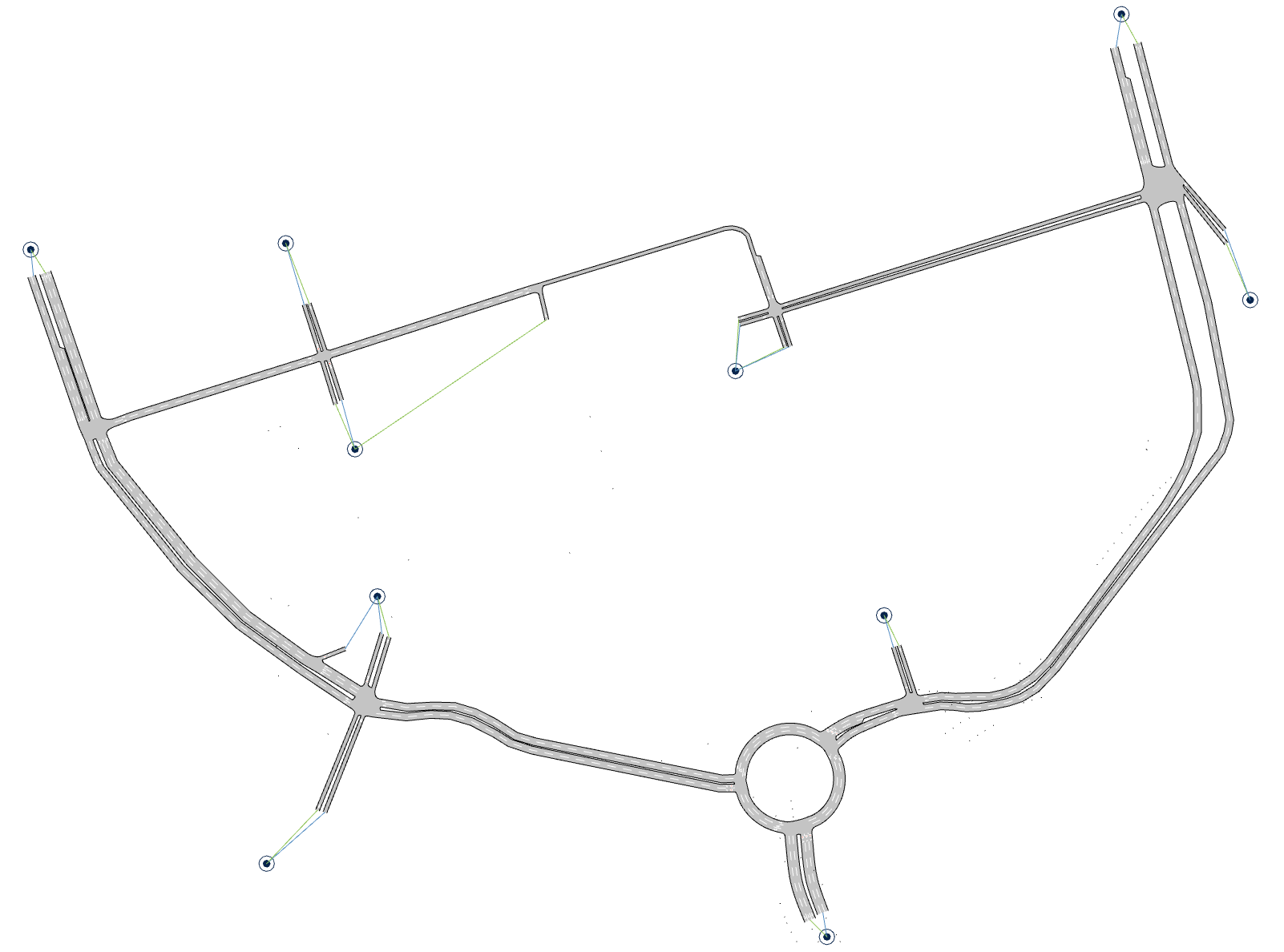}\qquad \quad
    \includegraphics[width=.45\textwidth]{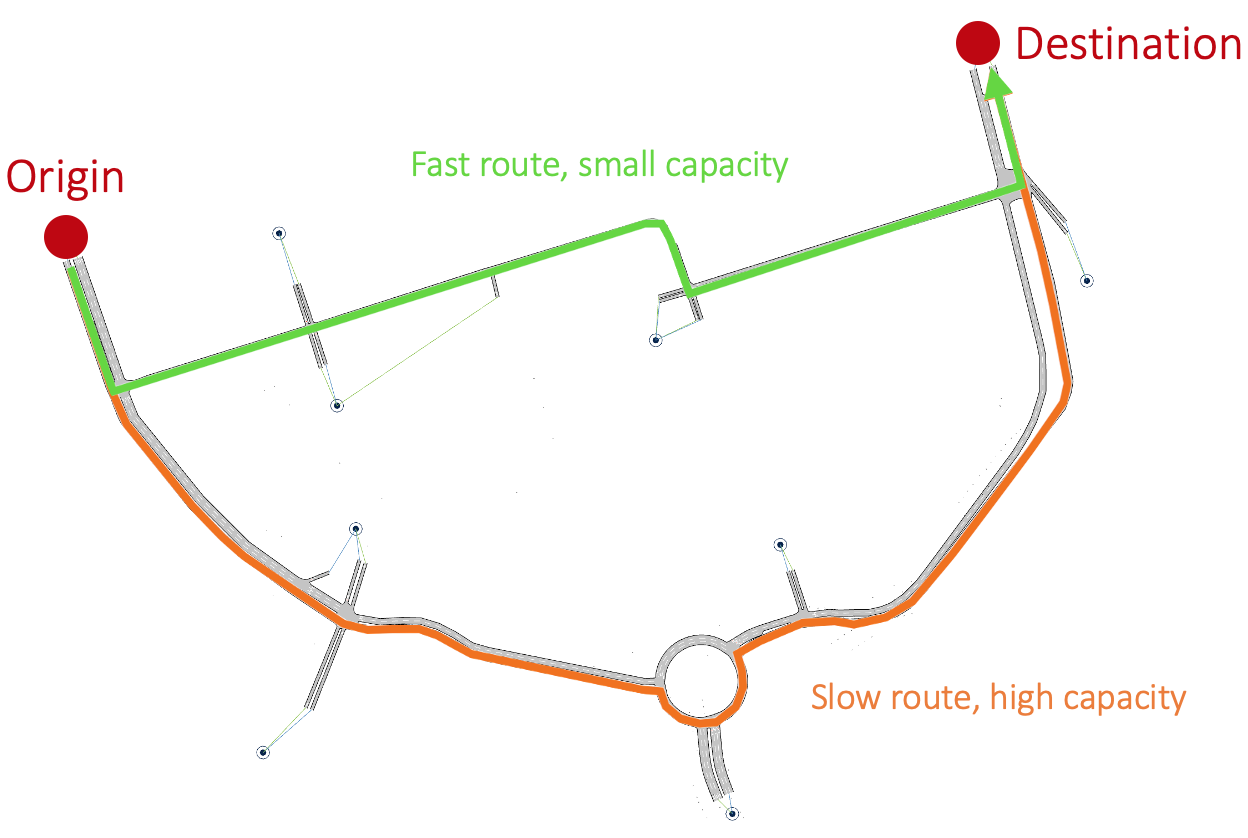}
    \caption{Network used for the Aimsun simulation.}
    \label{aimgraph}
\end{figure*}
Aimsun simulates the movement of individual vehicles through a network, based on a car-following model. Just like in our theoretical study, we focus on how access to real-time information influences routing decisions by distinguishing between app users and non-app users. App users are modeled to follow the app's recommendations, favoring routes with the shortest travel times. Conversely, non-app users are assumed to follow predetermined routes. Therefore, informed users adapt their choice to real-time traffic conditions, while non-app users do not.

\begin{figure*}[!t]
    \centering
    \includegraphics[width=.45\textwidth]{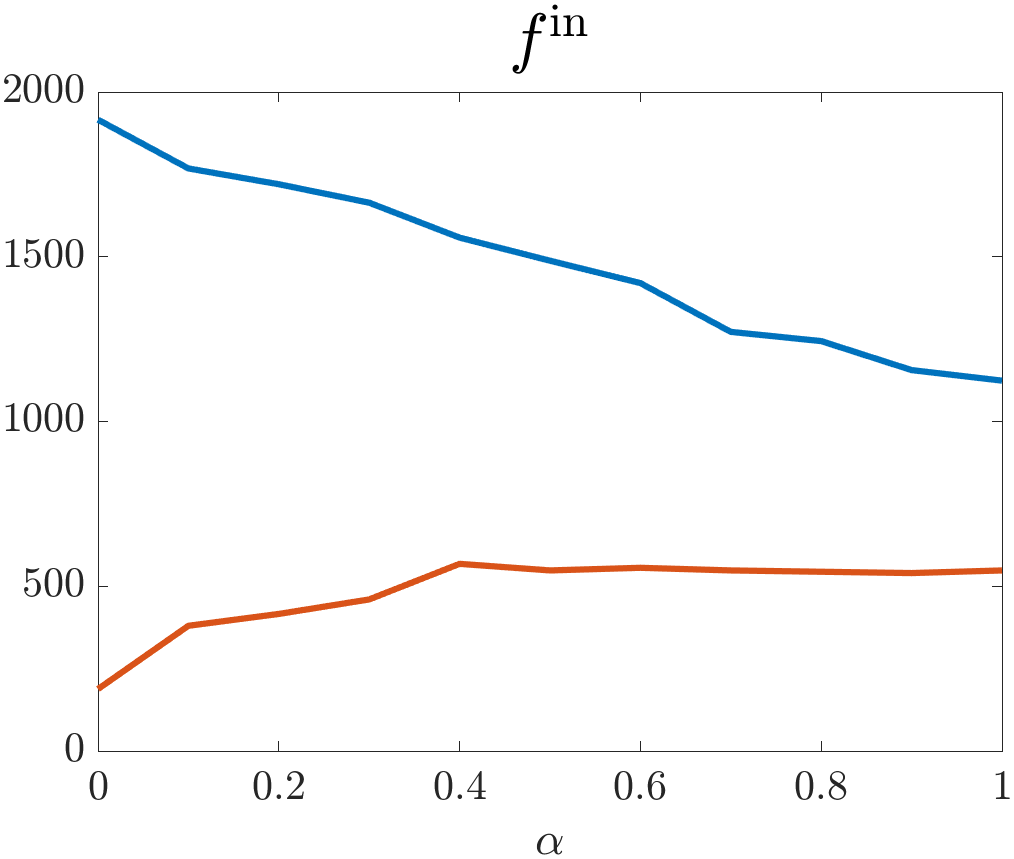}\qquad\includegraphics[width=.45\textwidth]{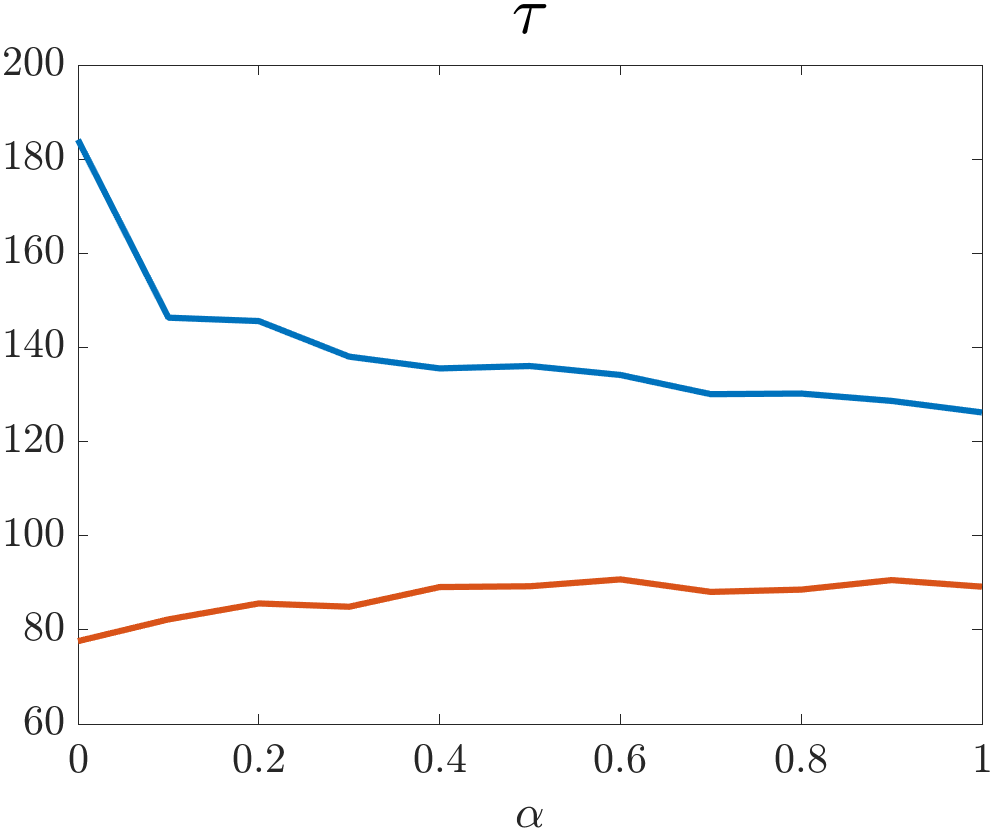}\\ \vspace{.5cm}
  \includegraphics[width=.32\textwidth]{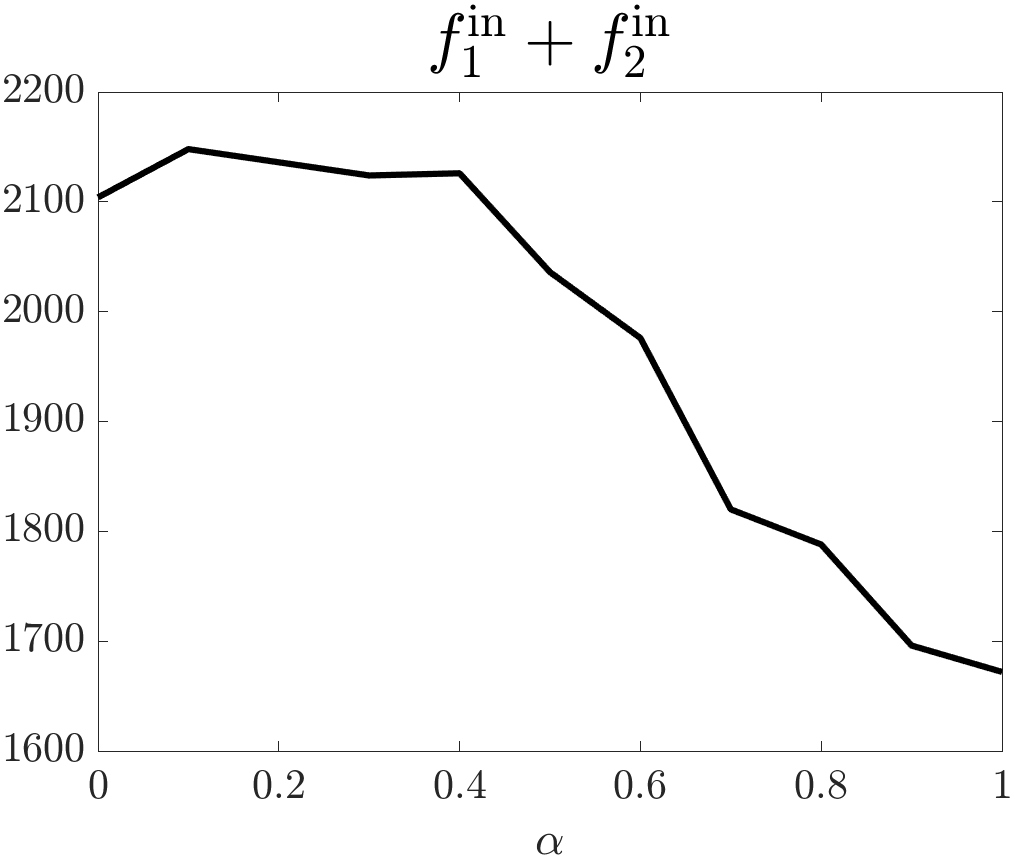}\ \ \includegraphics[width=.32\textwidth]{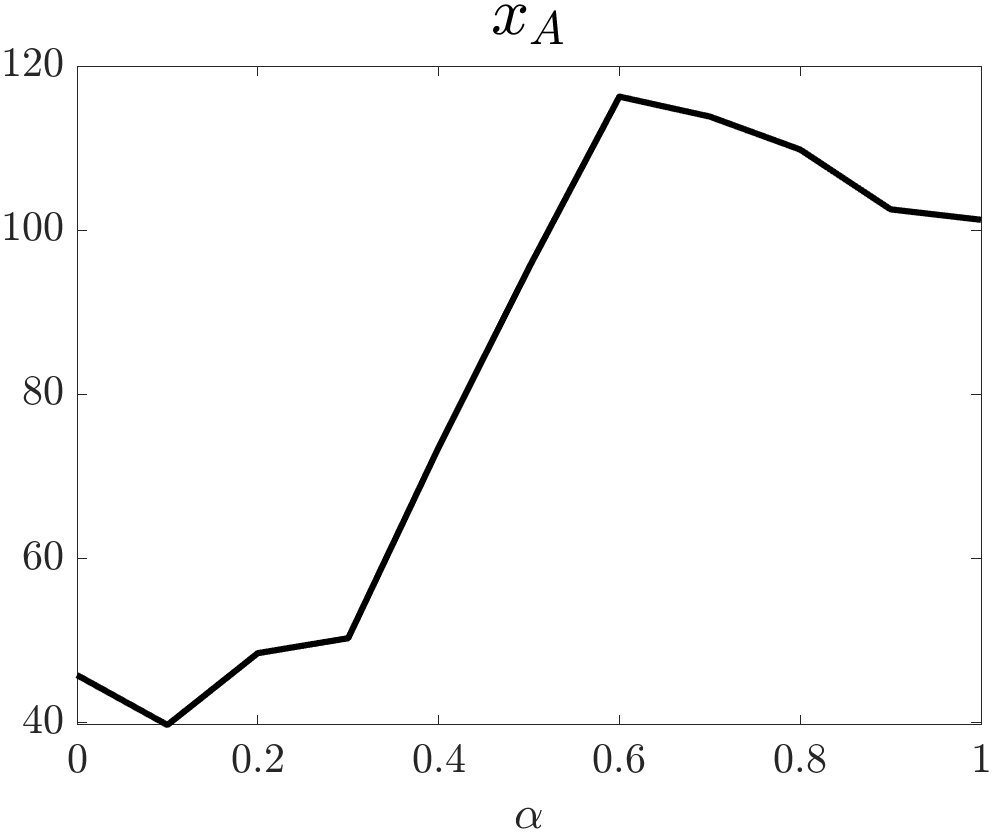}\ \ \includegraphics[width=.32\textwidth]{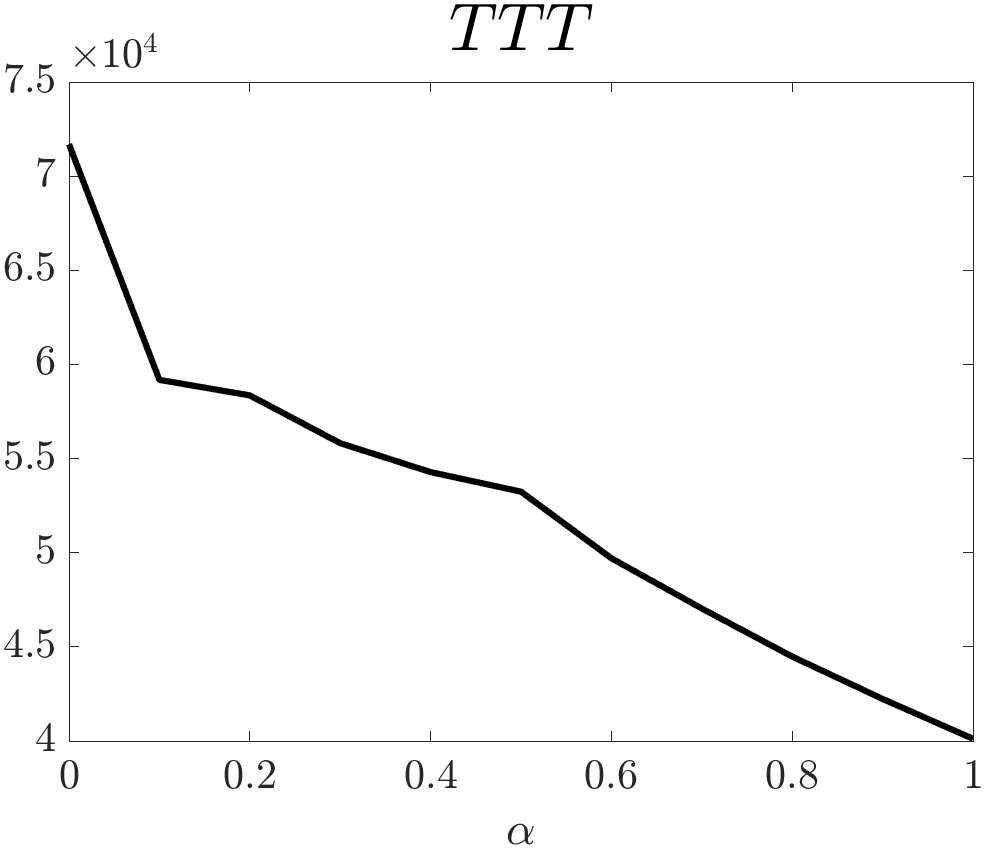}
    \caption{First line: time-averaged inflow to the two routes and time-averaged travel times on each route, each expressed as a function of the penetration rate, $\alpha$. Blue lines indicate the slow, high-capacity route, while orange lines represent the fast, small-capacity route. Data is derived from micro-simulations. Second line: time-averaged flow supplied to the two-route subsystem, time-averaged density on the access road/buffer, and total travel time across the two-route subsystem, all shown as functions of the penetration rate, $\alpha.$}
   \label{plot_aimsun}
\end{figure*}

The simulations were performed on the network shown in Figure~\ref{aimgraph}, featuring an origin-destination pair with two commuting options: a fast but low-capacity route (length: 0.875 km, maximum speed: 50 km/h, capacity of 900 veh/h) and a slower, high-capacity route (length: 1.35 km, maximum speed: 50 km/h, capacity of 1800 veh/h). The origin-destination pair is supplied with an exogenous flow $\Phi=2100$ veh/h through a short access road with a capacity of 2700 veh/h. Aimsun provides the possibility to directly model informed users using the logit choice model: in our simulations, informed users were modeled according to the logit choice model \eqref{lrr}, with a parameter of $\eta = 100$. This model was used to capture the decision-making process of app users.

Figure~\ref{plot_aimsun} shows the \emph{time-averaged} values of key quantities from the micro-simulations, i.e., their average value calculated over the entire simulation time window. In the first row, from left to right, we see the average inflow to the two routes and the average travel times on each route. The second row, from left to right, displays the time-averaged flow supplied to the two-route subsystem, the average density on the access road/buffer, and the total travel time across the two-route subsystem. We use the total travel time as a proxy for the Price of Anarchy (PoA) in this analysis, as the optimal value of the total travel time is not available from Aimsun simulations: this approximation provides a reasonable measure for comparing system efficiency under different routing scenarios, as the total travel time is the numerator of the PoA. These simulations were conducted for values of $\alpha$ ranging from 0 to 1 in increments of 0.1. For each simulation, the origin-destination pair was subjected to a traffic demand of $\Phi=2100$ vehicles, and the duration of each simulation was set to 15 minutes.

The plots confirm the theoretical findings and numerical experiments discussed in Section~\ref{sec5}. As the penetration rate of informed users increases, we observe a higher fraction of the exogenous flow being directed towards the fast, low-capacity route. However, when the penetration rate becomes excessively high, this route becomes saturated, leading to a reduction in the overall exogenous flow reaching the two-route sub-system. This saturation effect causes a corresponding rise in the average density on the access road. These results validate the idea that fastest-route recommendations can lead to partial transfer of demand, ultimately leading to congestion at the origin as reflected in the increased density. It should be noted that although these results align with the theoretical model in terms of the flow supplied to the two-road subsystem, the behavior of traffic density on the access road differs. Unlike our model, where density diverges to infinity when the system converges to a partially transferring equilibrium, the microscopic simulations show that density instead stabilizes at a finite value. This discrepancy is almost certainly due to the fact that the dynamics on the access road are not governed by a supply-and-demand mechanism. Finally, the plot of the total travel time is consistent with Figure~\ref{phi4} of the PoA before the inflow to the fast route saturates. It also shows that, for larger values of $\alpha$, the total travel time decreases even further under the partial demand transfer regime. This decrease is due to the reduced flow supplied to the two-route subsystem, which lessens congestion on the routes,  ultimately lowering the total travel time.

\section{Conclusion}\label{sec6}
In this work, we analyzed the impact of drivers' use of navigation apps on traffic through a dynamical network flow model, capturing the reactions of app-informed drivers to variations in travel times. The analysis focused on a network with two parallel links. We proved that the resulting dynamical system is globally asymptotically stable, allowing us to examine its equilibrium and gain insights into the effects of navigation apps on traffic efficiency. To achieve this, we modeled the behavior of app-informed drivers using the logit choice model, which enables us to analyze both high and low compliance with app recommendations. The results consistently highlighted that route recommendations can negatively impact traffic efficiency, going beyond the well-known issue of an increased Price of Anarchy. By incorporating a supply and demand mechanism on the network links, we were able to uncover a previously overlooked phenomenon: partial demand transfer.

This model can serve as a starting point for several meaningful extensions. For instance, the current model does not account for internal congestion. This limitation could be addressed by either incorporating routes composed of multiple Cell Transmission Model CTM-type links or by imposing a capacity limit at the destination node, which might lead to spillback on the two routes. Another important aspect would be to make the dynamics on the access roads more realistic, based on supply and demand mechanisms, to provide a more complete description of congestion phenomena at the network origin. 

More generally, extending the model to encompass more complex topologies with internal nodes and non-independent routes would enhance our understanding of the role of recommendations in traffic networks and potentially lead to a more systematic characterization of partial demand transfer. However, beyond the case studied in this paper, the properties of equilibrium uniqueness and monotonicity are not satisfied. Consequently, analyzing stability for these more complex networks poses challenges. Furthermore, as noted at the beginning of the chapter, the supply and demand constraints mean that many techniques used in previous works, which are based on the contractivity property \cite{como2013a, como2015, lovisari2014}, are not applicable in our case. Therefore, efforts will be directed toward identifying new techniques for stability analysis.


\section*{Acknowledgement}
This work has been supported in part by the French National Research Agency in the framework of the \emph{Investissements d’avenir} program ANR-15-IDEX-02 and of LabEx PERSYVAL ANR-11-LABX-0025-01.

\appendix

\section{Proof of Theorem~\ref{ch:dynselfish:t2}}\label{app1}
Theorem~\ref{ch:dynselfish:t2} is proved by exploiting the following result, based on monotonicity. While it was originally stated in \cite[Lemma 3]{mon} in terms of the \textit{semiflow} notion, we are going to state it for systems of differential equations.

\begin{defi}[Order and order convex intervals]
    Consider the partial order $\leq$ on $\mathcal{D}$: $x,y\in\mathcal{D}$ are ordered, i.e., $x\leq y$, if and only if $x_l\leq y_l,\ l=1,\dots,d$. Given two ordered points $x,y$, we define the \textit{order interval} \mbox{$[x,y]=\{p\in\mathcal{D}:x\leq p\leq y\}$}. A set $\mathcal{Y}$ is called \textit{order convex} if $[p,q]\subset\mathcal{Y}$, for every $p,q\in\mathcal{Y}$ such that $p\leq q$. Notice that order intervals are order convex sets.
\end{defi}

\begin{lem}[Monotonicity and stability]
    Consider a globally Lipschitz system $\dot{y}=h(y)$ with $h:D\rightarrow\mathbb{R}^d$ and $D\subset\mathbb{R}^d$. Suppose that:
    \begin{itemize}
        \item the system is monotone on $D$;
        \item the system admits a unique equilibrium $\overline{y}$ in $D$;
        \item every trajectory of the system has a compact closure.
        \item every neighborhood of every point $x\in D$ contains a compact and order convex neighborhood of $x$.
    \end{itemize}
    Then, the equilibrium $\overline{y}$ is globally asymptotically stable.
    \label{lem3}
\end{lem}
In order to apply Lemma \ref{lem3} to \eqref{td}, 
we need to verify the four conditions above. 
The third point is straighforward, after observing that $\Omega$ is compact. 
The fourth point is also immediate, if we let $N=\{x\in\Omega:|x-\overline{x}|<r\}=\{x\in\mathbb{R}^2:|x-\overline{x}|<r\}\cap\Omega$, $r>0$. By taking $p,q\in N$ such that $x\in[p,q]$, it is clear that $[p,q]\subset N$.
In the rest of this section, we shall prove the two remaining points.

\subsection{Monotonicity}
 We consider the extension of system \eqref{td} to the whole positive orthant $\mathbb{R}_{\geq0}^2$ and we shall prove the stronger property that the extension of \eqref{td} is a monotone system. To this end, we first show that it satisfies the so called \textit{K-condition}. 
System \eqref{td} is said to satisfy the K-condition if, given $a,b\in \mathbb{R}_{\geq0}^2$ such that $a\leq b$, where the inequality is meant component-wise, and $a_l=b_l$, then $\Sigma_l(b)\geq \Sigma_l(a)$.  
To verify the K-condition, notice that also for the extended systems we can identify the same system modes $\Sigma^{\text{M}_1\text{-M}_2}$ according to \eqref{td}, with the only difference that the corresponding regions $\overline{\Omega}^{\text{M}_1\text{-M}_2}$ can be unbounded and their union covers $\mathbb{R}_{\geq0}^2$  (the corresponding regions for the original system being their restrictions to $\Omega$. )
In particular, Assumption \ref{ch:dynselfish:ass4} guarantees that the jacobian matrices $J^{\text{M}_1\text{-M}_2}$ are Metzler. Then, K-condition holds inside every region $\overline{\Omega}^{\text{M}_1\text{-M}_2}$. 
The K-condition holds across different regions. To see this, consider $p,q\in\mathbb{R}_{\geq0}^2$ such that $p_1=q_1,\ p_2<q_2$ and $p$ and $q$ belong to different mode regions. We need to prove that $\Sigma_1(p)\leq\Sigma_1(q)$. From Figure~\ref{plotty}, one can see that we only need to compare points belonging to regions such that $p\in\overline{\Omega}^{\text{SF-M}_2}$, $q\in\overline{\Omega}^{\text{UF-M}_2}$, or $p\in\overline{\Omega}^{\text{SC-M}_2}$, $q\in\overline{\Omega}^{\text{UC-M}_2}$. 
\begin{itemize}
    \item $p\in\overline{R}^{\text{SF-M}_2}$, $q\in\overline{\Omega}^{\text{UF-M}_2}$:
    \begin{equation*}
        \Sigma_1(p)=\frac{\Phi R_1(\tau)-v_1p_1}{L_1}\leq \frac{F_1-v_1q_1}{L_1}= \Sigma_1(q);
    \end{equation*}
    \item $p\in\overline{\Omega}^{\text{SC-M}_2}$, $q\in\overline{\Omega}^{\text{UC-M}_2}$:
    \begin{equation*}
    \begin{aligned}
            \Sigma_1(p)&=\frac{\Phi R_1(\tau)-F_1}{L_1}\leq \frac{w_1(B_1-p_1)-F_1}{L_1} =\Sigma_1(q);
    \end{aligned}
    \end{equation*}
\end{itemize}
The inequalities above follow from the route mode's definition and by symmetry they also apply to Route~$2$, i.e., to the case in which $p,q\in\mathbb{R}_{\geq0}^2$ such that $p_1<q_1,\ p_2=q_2$, then we get that \eqref{td} satisfies to the K-condition.

\begin{figure}
        \centering
        \includegraphics[width=.8\textwidth]{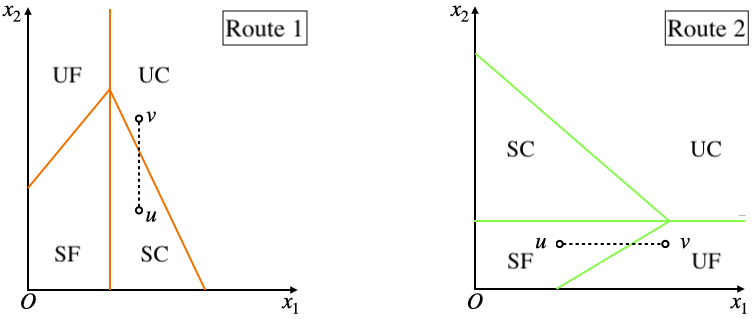}
        \caption{The K-condition can be easily verified, given a pair of points $u$ and $v$, considering only one of the two routes.}
        \label{plotty}
\end{figure}

We now prove that system \eqref{td} is monotone. 
Let us take $u,w\in\mathbb{R}_{\geq0}^2$ such that $u\leq w$. Let us take $\delta>0$ and define $\Sigma_\delta(x):=\Sigma(x)+\delta p$, where $p>\mathbf{0}$. Let $y_{\delta}(t)$ be the solution to the Cauchy problem
    \begin{equation*}
    \begin{cases}
        \dot{y}_{\delta}(t)=\Sigma_\delta(y_{\delta}(t))\\
        y_\delta(0)=w+\delta p,
    \end{cases}
    \end{equation*}
    and let $x(t,u)$ be the solution to \eqref{td} with initial condition $x(0)=u$. Notice that since $\Sigma_\delta(x)$ is Lipschitz, existence and uniqueness of $y_{\delta}(t)$ are guaranteed. Clearly, $x(0)<y_\delta(0)$. Our claim is that $x(t,u)<y_{\delta}(t),\  \forall t\geq 0$. By contradiction, suppose that there exist $\tau>0$ and $i$ such that
    \begin{equation*}
        x(\tau,u)\leq y_{\delta}(\tau),\quad x_l(\tau,u)=(y_{\delta}(\tau))_l, 
    \end{equation*}
    for some $i$, and $x(s,u)<y_{\delta}(s)$, $\forall s<\tau$. Then, it must hold that $\dot{x}_l(\tau,u)\geq(\dot{y}_\delta(\tau))_l$, which is equivalent to
    \begin{equation*}
        \Sigma_l(x(\tau,u))\geq (\Sigma_\delta)_l(y_\delta(\tau))=\Sigma_l(y_\delta(\tau))+\delta p_l.
    \end{equation*}
    This is absurd, since \eqref{td} satisfies to the K-condition, which implies that
    \begin{equation*}
        \Sigma_l(x(\tau,u))\leq (\Sigma_\delta)_l(y_\delta(\tau))=\Sigma_l(y_\delta(\tau))+\delta p_l.
    \end{equation*}
    Hence, $x(t,u)<y_{\delta}(t), \forall t\geq 0$, is proved. Finally, observe that, since $\Sigma$ and $\Sigma_\delta$ are Lipschitz continuous, by letting $\delta\rightarrow0$, the continuous dependence from initial conditions ensures that $x(t,u)\leq x(t,w),\ \forall t\geq 0.$
    We conclude by noting that the monotonicity property that we proved also holds true also for \eqref{td} restricted to $\Omega$, since $\Omega$ is an invariant set.

\subsection{Uniqueness of equilibrium}
Since $P$ is globally attractive, no equilibrium is contained in $Q$. This implies that the equilibria of \eqref{td} are contained in $P$. Hence, we can limit ourselves to only consider the system modes contained in $P$: 

        \begin{equation}
        \Sigma^{\text{SF-SF}}=\begin{cases}
            \dot{x}_1=\Phi R_1(\tau(x))-v_1x_1\\ \dot{x}_2=\Phi R_2(\tau(x))-v_2x_2
        \end{cases}
        \label{sfsf}
        \end{equation}

        \begin{equation}
          \Sigma^{\text{UF-SF}}=\begin{cases}
            \dot{x}_1=F_1-v_1x_1\\ \dot{x}_2=\Phi R_2(\tau(x))-v_2x_2
        \end{cases}
        \label{ufsf}
        \end{equation}

        \begin{equation}
          \Sigma^{\text{SF-UF}}=\begin{cases}
            \dot{x}_1=\Phi R_1(\tau(x))-v_1x_1\\ \dot{x}_2=F_2-v_2x_2
        \end{cases}
        \label{sfuf}
        \end{equation}
The equilibria of system \eqref{td} must coincide with the set of \textit{active} equilibria of the sub-systems above, 
where by active we mean that, if $x^{\Sigma^{\text{M}_1\text{-M}_2}}$ is an equilibrium of $\Sigma^{\text{M}_1\text{-M}_2}$, $x^{\Sigma^{\text{M}_1\text{-M}_2}}\in\overline{\Omega}^{{\text{M}_1\text{-M}_2}}$. We prove that the equilibrium of \eqref{td} is unique by showing that each of the above sub-systems admits a unique equilibrium and only one among them is active.

Consider the following sets:
    \begin{equation}
        \begin{aligned}
         &\gamma_l=\{x\in \Omega:\Phi R_l(\tau(x))-v_lx_l=0\},\quad l=1,2,\\
            &\epsilon_l=\{x\in \Omega:\Phi R_l(\tau(x))-F_l=0\},\quad l=1,2.
        \end{aligned}
        \label{ir}
    \end{equation}
    Sets $\gamma_1$, $\gamma_2$ are always nonempty. In fact, consider $x\in\Omega$ and notice that $x\in\gamma_l$ if and only if $x_l=\Phi R_l(\tau(x)))/v_l$. Since $R_l(\tau(x))$ is strictly decreasing in $x_l$ by Assumption~\ref{ch:dynselfish:ass4}, one can observe that for each $x_j\in[0,B_j]$, there exists a unique $x_l\in[0,B_l]$ satisfying to the above equality. The fact that $x_l\in[0,B_l]$ is guaranteed by Assumption \ref{ass6}. This proves that $\gamma_l\neq\emptyset,\ l=1,2$. 
    As for sets $\epsilon_1,\ \epsilon_2$, from \eqref{ch:dynselfish:rr} we have 
    \begin{equation*}
        r_l(x)=\frac{1}{\alpha}\left(\frac{F_l}{\Phi}-(1-\alpha)r_l^0\right),\quad l=1,2.
    \end{equation*}
    Since $0\leq r_l(x)\leq 1, \forall x\in\Omega,\ l=1,2$, in order to have $\epsilon_l\neq\emptyset$, it must be that
    \begin{equation*}
        (1-\alpha)\Phi r_l^0\leq F_l\leq (1-\alpha)\Phi r_l^0+\alpha\Phi.
    \end{equation*}
    
     Now, since routing ratios are assumed to be $C^1$ in $\Omega$ and strictly monotone by Assumption \ref{ch:dynselfish:ass4}, when the sets above are nonempty, the implicit function theorem guarantees the existence of functions between the variables $x_1$ and $x_2$ making explicit each of the implicit relations in \eqref{ir}. 
     Let the following be the functions associated to relations \eqref{ir}:
    \begin{equation}
    \begin{aligned}
        &x_2:=\Tilde{\gamma}_1(x_1),\ \ x_2:=\Tilde{\gamma}_2(x_1),\\ &x_2:=\Tilde{\epsilon}_1(x_1),\ \ x_2:=\Tilde{\epsilon}_2(x_1). 
    \end{aligned}
        \label{er}
    \end{equation}
    Then, sets in \eqref{ir} are the graphs of functions~\eqref{er}. The implicit function theorem and Assumption \ref{ch:dynselfish:ass4} also ensure that they are strictly increasing in $x_1$:
    \begin{equation*}
    \begin{aligned}
        &\Tilde{\gamma}'_1(x_1)=-\frac{\frac{\partial}{\partial x_1}(\Phi R_1(\tau(x)-v_1x_1)}{\frac{\partial}{\partial x_2}(\Phi R_1(\tau(x))-v_1x_1)}=\frac{v_1-\frac{\partial R_1(\tau(x))}{\partial x_1}}{\frac{\partial R_1(\tau(x))}{\partial x_2}}>0,
        \end{aligned}
    \end{equation*}
    \begin{equation*}
    \begin{aligned}
    &\Tilde{\gamma}'_2(x_1)=-\frac{\frac{\partial}{\partial x_1}(\Phi R_2(\tau(x)-v_2x_2)}{\frac{\partial}{\partial x_2}(\Phi R_2(\tau(x))-v_2x_2)}=\frac{\frac{\partial R_2(\tau(x))}{\partial x_1}}{v_2-\frac{\partial R_2(\tau(x))}{\partial x_2}}>0,
    \end{aligned}
    \end{equation*}
    and
    \begin{equation*}
        \Tilde{\epsilon}'_l(x_l)=-\frac{\frac{\partial}{\partial x_1}(\Phi R_l(\tau(x))-F_l)}{\frac{\partial}{\partial x_2}(\Phi R_l(\tau(x))-F_l)}=-\frac{\frac{\partial R_l(\tau(x))}{\partial x_1}}{\frac{\partial R_l(\tau(x))}{\partial x_2}}>0,
    \end{equation*}
for $l=1,2.$
    
   We now show that each sub-system in $P$ admits a unique equilibrium point. First of all, observe that the image of $\Tilde{\gamma}_1(x_1)$ is the whole $[0,B_2]$, since, for each $x_2\in[0,B_2]$, the map $x_1=\Phi R_1(\tau_1(x_1),\tau_2(x_2))/v_1$ always admits a unique fixed point in $[0,B_1]$. Analogously, $\Tilde{\gamma}_2(x_1)$ is defined over the whole interval $[0,B_1]$, since, for each $x_1\in[0,B_1]$, the map $x_2=\Phi R_1(\tau_1(x_1),\tau_2(x_2))/v_2$ always admits a unique fixed point $[0,B_2]$. This, combined with the fact that both functions are strictly increasing in $x_1$, implies that $\Tilde{\gamma}_1(x_1)$ and $\Tilde{\gamma}_2(x_1)$ intersect lines $\{x_2=C_2\}$ and$\{x_1=C_1\}$ in one and only one point inside $\Omega$, respectively. These two points correspond to the unique equilibria of systems $\Sigma^{\text{SF-UF}}$, $\Sigma^{\text{UF-SF}}$, namely $x^{\Sigma^{\text{SF-UF}}}$, $x^{\Sigma^{\text{UF-SF}}}$.  

    Now, from above, it follows that the functions $\Tilde{\gamma}_1(x_1)$ and $\Tilde{\gamma}_2(x_1)$ admit two points of the form $(x_1,0)$, $(0,x_2)$ satisfying to their equations, respectively. Moreover, we have that
    \begin{equation*}
    \begin{aligned}
        \Tilde{\gamma}'_2(x_1)-\Tilde{\gamma}'_1(x_1)
        &=-\frac{v_1v_2+\sum_{i\neq j}v_l\frac{\partial R_j(x)}{\partial x_l}}{\frac{\partial R_1(\tau(x))}{\partial x_2}\left(v_2+\frac{\partial R_1(\tau(x))}{\partial x_2}\right)}\\
        &\leq-\frac{v_1v_2}{K(v_2+K)}<0.
        \end{aligned}
    \end{equation*}
    The combination of these two facts implies that
    there exist a unique point in $\Omega$ such that $\Tilde{\gamma}_1(x_1)=\Tilde{\gamma}_2(x_1)$, which corresponds to the unique equilibrium point of sub-system $\Sigma^{\text{SF-SF}}$, namely $x^{\Sigma^{\text{SF-SF}}}$.

    Finally, we prove one and only one among $x^{\Sigma^{\text{SF-SF}}}$, $x^{\Sigma^{\text{UF-SF}}}$, $x^{\Sigma^{\text{SF-UF}}}$ is active. Before proceeding, let us point out some facts. First of all, observe that when $\epsilon_1$ and $\epsilon_2$ are non empty, the graph of the their functions $\Tilde{\epsilon}_1(x_1)$, $\Tilde{\epsilon}_2(x_1)$ split the state space into two separate regions:
    \begin{equation*}
    \begin{aligned}
        &E_l^+:=\{x\in\Omega:\ \Phi R_l(\tau(x))-F_l>0\},\\ 
        &E_l^-:=\{x\in\Omega:\ \Phi R_l(\tau(x))-F_l<0\},
    \end{aligned}
    \end{equation*}
    $l=1,2$. Regions $E_1^+$ and $E_2^+$ are those where there is unsatisfied demand on route $1$ and $2$, respectively. 
    One can see that $E_1^+$ and $E_1^-$ stand above and below the graph of $\Tilde{\epsilon}_1(x_1)$, respectively, whereas $E_2^+$ and $E_2^-$ stand below and above the graph of $\Tilde{\epsilon}_2(x_2)$, respectively. Observe that the following identities hold: 
    \begin{center}
    $\Omega^{\text{SF-SF}}=E_1^-\cap E_2^-\cap P,\quad\Omega^{\text{UF-SF}}=E_1^+\cap E_2^-\cap P$,\\ $ $\\ $\Omega^{\text{SF-UF}}=E_1^-\cap E_2^+\cap P$.
    \end{center}
     Assumptions~\ref{ass1} and \ref{ch:dynselfish:ass4} ensure that $\Tilde{\epsilon}_1(x_1)>\Tilde{\epsilon}_2(x_1),\ \forall x\in\Omega$ when both functions are defined, and $\Tilde{\epsilon}_1(x_1)\geq\Tilde{\gamma}_1(x_1), \ x_1\leq C_1$, $\Tilde{\epsilon}_2(x_1)\leq\Tilde{\gamma}_2(x_1), \ x_2\leq C_2$, where the equality holds if and only if $x_1=C_1$ and $x_2=C_2$, respectively. Also, recall that $\Tilde{\gamma}_2(x_1)\geq \Tilde{\gamma}_1(x_1),\ x_1\leq \left(x^{\text{SF-SF}}\right)_1$, where the equality holds if and only if $x_1= \left(x^{\text{SF-SF}}\right)_1$, and $\Tilde{\gamma}_2(x_1)< \Tilde{\gamma}_1(x_1),\ x_1> \left(x^{\text{SF-SF}}\right)_1$.
    We are going to distinguish two different cases.

{\textbf{Case 1:} $x^{\Sigma^{\text{SF-SF}}}\in P$.}
        In order for $x^{\Sigma^{\text{SF-SF}}}$ to be active, it must hold that $x^{\Sigma^{\text{SF-SF}}}\in \Omega^{\text{SF-SF}}$. This leads to two distinct sub-cases.
        \begin{itemize}
            \item \textbf{Sub-case 1:} $\epsilon_1$, $\epsilon_2$ are both empty.\\
            From Assumption \ref{ass1}, $\epsilon_1$, $\epsilon_2$ are both empty if and only $F_l\geq\alpha\Phi+(1-\alpha)\Phi r_l^0,\ l=1,2$. Hence, unsatisfied demand cannot arise on either route inside $\Omega$. Thus, $P\equiv\overline{\Omega}^{\text{SF-SF}}$. The latter implies that neither $x^{\Sigma^{\text{UF-SF}}}$ nor $x^{\Sigma^{\text{SF-UF}}}$ can be active, hence $\overline{x}\equiv x^{\Sigma^{\text{SF-SF}}}$.
            \item \textbf{Sub-case 2:} at least one among $\epsilon_1$, $\epsilon_2$ is non empty. \\
            In this case, region $P$ undergoes a partition. Nevertheless, the fact that $\Tilde{\epsilon}_1(x_1)>\Tilde{\epsilon}_2(x_1)$, $\Tilde{\epsilon}_1(x_1)\geq\Tilde{\gamma}_1(x_1)$, $\Tilde{\epsilon}_2(x_1)\leq\Tilde{\gamma}_2(x_1)$ ensures that $\Omega^{\text{SF-SF}}$ is non-empty and that $x^{\Sigma^{\text{SF-SF}}}\in \Omega^{\text{SF-SF}}$. Thus, $\overline{x}\equiv x^{\Sigma^{\text{SF-SF}}}$. As for $x^{\Sigma^{\text{UF-SF}}}$, it cannot be active, since the intersection between the graph of $\Tilde{\gamma}_2(x_1)$ and the line $\{x_1=C_1\}$ occurs below the graph of $\Tilde{\epsilon}_1(x_1)$. The latter follows from $\Tilde{\gamma}_2(x_1)\leq \Tilde{\gamma}_1(x_1),\ x_1> \left(x^{\text{SF-SF}}\right)_1$. Analogously, $x^{\Sigma^{\text{SF-UF}}}$ cannot be active.
        \end{itemize}
        
        \textbf{Case 2:} $x^{\Sigma^{\text{SF-SF}}}\notin P$.
        First of all, notice that if $x^{\Sigma^{\text{SF-SF}}}\notin P$, then it is not active and either $x^{\Sigma^{\text{SF-SF}}}\in \{x_1>C_1,\ x_2\leq C_2\}$ or $x^{\Sigma^{\text{SF-SF}}}\in \{x_1\leq C_1,\ x_2> C_2\}$. Indeed, if $x^{\Sigma^{\text{SF-SF}}}\in \{x_1>C_1,\ x_2> C_2\}$, then $\Phi R_l(\tau(x^{\Sigma^{\text{SF-SF}}}))=v_l(x^{\Sigma^{\text{SF-SF}}})_l>F_l,\ l=1,2$,
        contradicting Assumption~\ref{ass1}. Suppose then that $x^{\Sigma^{\text{SF-SF}}}\in \{x_1>C_1,\ x_2\leq C_2\}$. Then, it must be that $\Tilde{\gamma}_2(C_1)> \Tilde{\gamma}_1(C_1)$ and $\Tilde{\gamma}_2(C_1)\leq C_2$. Moreover, $\Tilde{\epsilon}_1(C_1)= \Tilde{\gamma}_1(C_1)$. Thus, $x^{\Sigma^{\text{UF-SF}}}\in \Omega^{\text{UF-SF}}$. Finally, observe that, again, since $\Tilde{\gamma}_2(C_1)> \Tilde{\gamma}_1(C_1)$ and $\Tilde{\gamma}_2(C_1)\leq C_2$, $\Tilde{\gamma}_1(x_1)$ cannot intersect the line $\{x_2=C_2\}$ inside $P$, i.e, $x^{\Sigma^{\text{SF-UF}}}$ is not active.
        One can repeat the same process in the case in which $x^{\Sigma^{\text{SF-SF}}}\in \{x_1\leq C_1,\ x_2> C_2\}$. In this case, $x^{\Sigma^{\text{SF-UF}}}$ is the only active equilibrium point.
    
    To conclude, it might be that $x^{\Sigma^{\text{SF-SF}}}$ coincides with $x^{\Sigma^{\text{UF-SF}}}$ or $x^{\Sigma^{\text{SF-UF}}}$. One can verify that this only happens when $x^{\Sigma^{\text{SF-SF}}}\in\epsilon_1$ and $x^{\Sigma^{\text{SF-SF}}}\in\epsilon_2$, respectively. In this case, the two coinciding equilibria represent the unique equilibrium $\overline{x}$ of the system .
    


\end{document}